\documentclass[11pt]{article}
\usepackage{amsmath}
\usepackage{amsthm}
\usepackage{amssymb}
\usepackage{algorithm}
\usepackage{subfig}
\usepackage{color}
\usepackage[english]{babel}
\usepackage{graphicx}
\usepackage{wrapfig,epsfig}
\usepackage{epstopdf}
\usepackage{url}
\usepackage{graphicx}
\usepackage{color}
\usepackage{epstopdf}
\usepackage{algpseudocode}
\usepackage{scrextend}
\usepackage[T1]{fontenc}
\usepackage{bbm}
\usepackage{comment}
\usepackage{authblk}
\usepackage{multicol}
\usepackage{dsfont}

\let\C\relax
\usepackage{tikz}
\usepackage{hyperref}  
\hypersetup{colorlinks=true,citecolor=blue,linkcolor=blue} 
\usetikzlibrary{arrows}
\usepackage[margin=1in]{geometry}
\graphicspath{{./figs/}}

\author{
  Vasileios Nakos, Zhao Song, Zhengyu Wang \\
  \texttt{ \{vasileiosnakos,zhaos,zhengyuwang\}@g.harvard.edu} \\
  Harvard University
}  

\date{}

\title{(Nearly) Sample-Optimal Sparse Fourier Transform in Any Dimension; \textsc{RIP}less and Filterless}

\newtheorem{theorem}{Theorem}[section]
\newtheorem{lemma}[theorem]{Lemma}
\newtheorem{definition}[theorem]{Definition}

\newtheorem{fact}[theorem]{Fact}
\newtheorem{remark}[theorem]{Remark}

\newcommand{\wh}{\widehat}
\newcommand{\wt}{\widetilde}
\newcommand{\ov}{\overline}

\renewcommand{\i}{\mathbf{i}}

\renewcommand{\varepsilon}{\epsilon}
\renewcommand{\tilde}{\wt}

\renewcommand{\bar}{\ov}

\DeclareMathOperator*{\E}{{\bf {E}}}
\DeclareMathOperator*{\Var}{{\bf {Var}}}
\DeclareMathOperator*{\Z}{\mathbb{Z}}
\DeclareMathOperator*{\C}{\mathbb{C}}

\DeclareMathOperator{\supp}{supp}
\DeclareMathOperator{\sparse}{sparse}
\DeclareMathOperator{\poly}{poly}

\makeatletter
\newcommand*{\RN}[1]{\expandafter\@slowromancap\romannumeral #1@}
\makeatother

\usepackage{lineno}

\usepackage{etoolbox}

\makeatletter

\newcommand{\define}[4][ignore]{%
  \ifstrequal{#1}{ignore}{}{
  \@namedef{thmtitle@#2}{#1}}%
  \@namedef{thm@#2}{#4}%
  \@namedef{thmtypen@#2}{lemma}%
  \newtheorem{thmtype@#2}[theorem]{#3}%
  \newtheorem*{thmtypealt@#2}{#3~\ref{#2}}%
}

\newcommand{\state}[1]{%
  \@namedef{curthm}{#1}
  \@ifundefined{thmtitle@#1}{
  \begin{thmtype@#1}
    }{
  \begin{thmtype@#1}[\@nameuse{thmtitle@#1}]
  }
    \label{#1}
    \@nameuse{thm@#1}
  \end{thmtype@#1}
  \@ifundefined{thmdone@#1}{
  \@namedef{thmdone@#1}{stated}%
  }{}
}

\newcommand{\restate}[1]{%
  \@namedef{curthm}{#1}
  \@ifundefined{thmtitle@#1}{
    \begin{thmtypealt@#1}
    }{
  \begin{thmtypealt@#1}[\@nameuse{thmtitle@#1}]
  }
    \@nameuse{thm@#1}
  \end{thmtypealt@#1}
  \@ifundefined{thmdone@#1}{
  \@namedef{thmdone@#1}{stated}%
  }{}
}

\newcommand{\thmlabel}[1]{
  \@ifundefined{thmdone@\@nameuse{curthm}}{\label{#1}
    }{\tag*{\eqref{#1}}}
}
\makeatother

\begin{document}

\begin{titlepage}
  \maketitle
  \begin{abstract}
In this paper, we consider the extensively studied problem of computing a $k$-sparse approximation to the $d$-dimensional Fourier transform of a length $n$ signal. Our algorithm uses $O(k \log k \log n)$ samples, is dimension-free, operates for any universe size, and achieves the strongest $\ell_{\infty}/\ell_2$ guarantee, while running in a time comparable to the Fast Fourier Transform. In contrast to previous algorithms which proceed either via the Restricted Isometry Property or via filter functions, our approach offers a fresh perspective to the sparse Fourier Transform problem.

  \end{abstract}
  \thispagestyle{empty}
\end{titlepage}

\newpage



\section{Introduction}

Initiated in discrete signal processing, compressed sensing/sparse recovery is an extensively studied branch of mathematics and algorithms, which postulates that a small number of linear measurements suffice to approximately reconstruct the best $k$-sparse approximation of a vector $x\in \mathbb{C}^n$ \cite{ct06,crt06,d06}. 
Besides substantial literature on the subject, compressed sensing has wide real-world applications in imaging, astronomy, seismology, etc. 
One of the initial papers in the field, Candes, Romberg and Tao \cite{crt06a}, has almost $15,000$ references.

Probably the most important subtopic is sparse Fourier transform, which aims to reconstruct a $k$-sparse vector from Fourier measurements. 
In other words, measurements are confined to the so-called Fourier ensemble. In Optics imaging \cite{g05,v11} and Magnetic resonance imaging (\textsc{MRI}) \cite{assn08}, the physics \cite{r89} of the underlying device restricts us to the Fourier ensemble, where the sparse Fourier problem becomes highly relevant. In fact, these applications became one of the inspirations for Candes, Romberg and Tao. 
The number of samples plays a crucial role: they determine the amount of radiation a patient receives in CT scans and taking fewer samples can reduce the amount of time the patient needs to spend in the machine. 
The framework has found its way in life-changing applications, including \textsc{Compressed Sensing GRAB-VIBE}, \textsc{CS SPACE}, \textsc{CS SEMAC} and \textsc{CS TOF} by Siemens \cite{siemens}, and Compressed Sense by Phillips \cite{phillips}. 
Its incorporation into the MRI technology allows faster acquisition rates, depiction of dynamic processes or moving organs, as well as acceleration of \textsc{MRI} scanning up to a factor of $40$. 
In the words of \textsc{SIEMENS} Healthineers: 

{ \emph{ This allows bringing the advantages of Compressed Sensing \textsc{GRASP-VIBE} to daily clinical routine.
\begin{itemize}
\item Perform push-button, free-breathing liver dynamics.
\item Overcome timing challenges in dynamic imaging and respiratory artifacts.
\item Expand the patient population eligible for abdominal \textsc{MRI}.
\end{itemize}
}}

On the other hand, Fourier transform is ubiquitous: image processing, audio processing, telecommunications, seismology, polynomial multiplication, \textsc{Subset Sum} and other textbook algorithms are some of the best-known applications of Fast Fourier Transform. 
The Fast Fourier Transform (FFT) by Cooley and Tukey \cite{ct65} runs in $O(n \log n)$ time and has far-reaching impact on all of the aforementioned cases. We can thus expect that algorithms which exploit sparsity assumptions about the input and can outperform FFT in applications are of high practical value. 
Generally, the two most important parameters one would like to optimize are sample complexity, i.e. the numbers needed to obtain from the time domain, as well as the time needed to approximate the Fourier Transform.

Two different lines of research exist for the problem: one focuses solely on sample complexity, while the other tries to achieve sublinear time while keeping the sample complexity as low as possible. 
The first line of research operates via the renowned Restricted Isometry Property (RIP), which proceeds by taking random samples and solving a linear/convex program, or an iterative thresholding procedure \cite{ct06,ddts06,tg07,bd08,dm08,rv08,bd09short,bd09,nt09,nv09,gk09,bd10,nv10,f11,b14,hr16}. Such algorithms are analyzed in two steps as follows: 
The first step ensures that, after sampling an appropriate number of points from the time domain, the inverse DFT matrix restricted on the rows indexed by those points acts as a near isometry on the space of $k$-sparse vectors. 
All of the state-of-the-art results \cite{ct06,rv08,b14,hr16} employ chaining arguments to make the analysis of this sampling procedure as tight as possible. 
The second part is to exploit the aforementioned near-isometry property to find the best $k$-sparse approximation to the signal. 
There, existing approaches either follow an iterative procedure which gradually denoise the signal \cite{bd08,nt09,nv09}, or perform $\ell_1$ minimization \cite{ct06}, a method that promotes sparsity of solutions.

The second line of research tries to implement arbitrary linear measurements via sampling Fourier coefficients \cite{gl89,m92,km93,ggims02,ags03,gms05,i08,i10,hikp12a,hikp12b,lawlor13,iw13,PR14,ikp14,ik14,k16,k17,ci17,bzi17,mzic2017,ln19} and use sparse functions (in the time domain) which behave like bandpass filters in the frequency domain. 
The seminal work of Kapralov \cite{k17} achieves $O(k \log n)$ samples and running time that is some $\log $ factors away from the sample complexity. 
This would be the end of the story, if not for the fact that this algorithm does not scale well with dimension, since it has an exponential dependence on $d$. 
Indeed, in many applications, one is interested in higher dimensions, rather than the one-dimensional case. 
The main reason\footnote{But not the only one: pseudorandom permutations for sparse FT in high dimensions also incur an exponential loss, and it is not known whether this can be avoided.} why this curse of dimensionality appears is the lack of dimension-independent ways to construct functions that approximate the $\ell_\infty$ ball and are sufficiently sparse in the time domain. 
A very nice work of Kapralov, Velingker and Zandieh \cite{kvz19} tries to remedy that by combining the standard execution of FFT with careful aliasing, but their algorithm works in a noiseless setting, and has a polynomial, rather than linear, dependence on $k$; the running time is polynomial in $k, \log n$ and the exponential dependence is avoided. 
It is an important and challenging question whether a robust and more efficient algorithm can be found. 

We note that in many applications, such as MRI or computed tomography (CT), the main focus is the sample complexity; the algorithms that have found their way to industry are, to the best of our knowledge, not concerned with sublinear running time, but with the number of measurements, which determine the acquisition time, or in CT the radiation dose the patient receives. 
Additionally, it is worth noting some recent works on sparse Fourier transform in the continuous setting, see \cite{i10,iw13,iw13,bcgls14,ps15,ckps16,s17,hkmmvz19,cp19a,cp19b,song19}.

\paragraph{Our Contribution.} We give a new algorithm for the sparse Fourier transform problem, which has $O( k \log n \log k)$ sample complexity for any dimension, and achieves the $\ell_\infty/\ell_2$ guarantee\footnote{This is the strongest guarantee in the sparse recovery literature. See also the caption of the table in Section 1.2}, while running in time $\tilde{O}(n)$. 
The previous state-of-the-art algorithm that achieved such a guarantee is the work of Indyk and Kapralov \cite{ik14}, which has $ 2^{O(d \log d)} k \log n $ sample complexity; an exponentially worse dependence on $d$. 
The work of \cite{hr16} obtains $O(k \log n \log^2k)$ samples in any dimension, but has a much weaker error guarantee\footnote{They achieve $\ell_2/\ell_1$ instead of $\ell_\infty/\ell_2$, see next Section for comparison. }, while their approach requires $\Omega( k \log n \log k)$ samples in high dimensions \cite{r19}. 
Moreover, the algorithm in \cite{ik14} operates when the universe size in each dimension is a power of $2$, whereas there is no restriction in our work. 
To obtain our result, we introduce a set of new techniques, deviating from previous work, which used the Restricted Isometry Property and/or filter functions. 




\subsection{Preliminaries}\label{sec:preliminary}

For any positive integer $n$, we use $[n]$ to denote $\{1,2,\cdots,n\}$. 
We assume that the universe size $n = p^d$ for any positive integer $p$. Our algorithm facilitates $n = \Pi_{j=1}^d p_j$ for any positive integers $p_1,\ldots,p_d$, but we decide to present the case $n=p^d$ for ease of exposition; the proof is exactly the same in the more general case.
Let $\omega = e^{2\pi\i / p}$ where $\i = \sqrt{-1}$.
We will work with the normalized $d$-dimensional Fourier transform
\begin{align*}
\wh{x}_f = \frac{1}{ \sqrt{n} } \sum_{t \in [p]^d} x_t \cdot \omega^{f^\top t}, & ~ \forall f \in [p]^d
\end{align*}
and the inverse Fourier transform is
\begin{align*}
x_t = \frac{1}{ \sqrt{n} } \sum_{f \in [p]^d} \wh{x}_f \cdot \omega^{-f^\top t}, & ~ \forall t \in [p]^d.
\end{align*}

For any vector $x$ and integer $k$, we denote $x_{-k}$ to be the vector obtained by zeroing out the largest (in absolute value) $k$ coordinates from $x$.

\subsection{Our result}

Apart from being dimension-independent and working for any universe size, our algorithm satisfies $\ell_\infty/\ell_2$, which is the strongest guarantee out of the standard guarantees considered in compressed sensing tasks. A guarantee $G_1$ is stronger than guarantee $G_2$ if for any $k$-sparse recovery algorithm that satisfies $G_1$ we can obtain a $O(k)$-sparse recovery algorithm that satisfies $G_2$. See also below for a comparison between $\ell_\infty/\ell_2$ and $\ell_2/\ell_2$, the second stronger guarantee.

Previous work is summarized in Table \ref{tbl:table}. Our result is the following.

\begin{theorem}[main result, informal version]\label{thm:fourier_sparse_recovery_informal}
Let $n = p^d$ where both $p$ and $d$ are positive integers.
Let $x \in \C^{[p]^d}$.
Let $k \in \{1, \ldots, n\}$. 
Assume that $R^* \geq \|\wh{x}\|_\infty / \|\wh{x}_{-k}\|_2$ where $\log R^* = O(\log n)$ (signal-to-noise ratio).
There is an algorithm that takes $O(k \log k \log n)$ samples from $x$, runs in $\wt{O}( n )$ time, and outputs a $O(k)$-sparse vector $y$ such that
\begin{align*}
\| \wh{x} - y \|_{\infty} \leq \frac{1}{\sqrt{k}} \|\wh{x}_{-k}\|_2
\end{align*}
holds with probability at least $1 - 1 / \poly(n)$.
\end{theorem}

\begin{table}[!t]
\centering
  \begin{tabular}{ | l | l | l | l | l | l|}
    \hline
    {\bf Reference} & {\bf Samples} & {\bf Time} & {\bf Filter} & {\bf RIP} & {\bf Guarantee} \\ \hline
    \cite{gms05} & $k \log^{O(d)}n$ & $k \log^{ O(d)} n$ & Yes & No & $\ell_2/\ell_2$ \\ \hline
    \cite{ct06}& $k\log^6 n$ & $\poly(n)$ & No & Yes & $\ell_2/\ell_1$ \\ \hline
    \cite{rv08} & $k \log^2 k \log(k \log n) \log n$ & $\tilde{O}(n)$ & No & Yes & $\ell_2/\ell_1$\\ \hline
    \cite{hikp12a} & $k \log^d n \log (n/k)$ & $k \log^d n \log (n/k)$ & Yes & No & $\ell_2/\ell_2$ \\ \hline
    \cite{cgv13} & $k\log^3 k \log n$ & $\tilde{O}(n)$ & No & Yes & $\ell_2/\ell_1$\\ \hline
    \cite{ik14} & $2^{d \log d} k \log n $ & $\tilde{O}( n)$ & Yes & No & $\ell_{\infty}/\ell_2$ \\ \hline
    \cite{b14} & $k \log k \log^2 n$ & $\tilde{O}(n)$ & No & Yes & $\ell_2/\ell_1$ \\ \hline
    \cite{hr16} & $k\log^2 k \log n$ & $\tilde{O}(n)$ & No & Yes & $\ell_2/\ell_1$ \\ \hline
    \cite{k16,k17} & $2^{d^2} k \log n  $ & $2^{d^2}k \log^{d+O(1)} n $ & Yes & No & $\ell_2/\ell_2$ \\ \hline
    \cite{kvz19} & $k^3 \log^2 k \log^2 n$ & $k^3 \log^2 k \log^2 n$ & Yes & Yes & Exactly $k$-sparse\\ \hline
    Theorem~\ref{thm:fourier_sparse_recovery_informal} & $k \log k \log n$ & $\tilde{O}(n)$ & No & No & $\ell_{\infty}/\ell_2$ \\ \hline
  \end{tabular}\caption{$n = p^d$. We ignore the $O$ for simplicity. The $\ell_{\infty}/\ell_2$ is the strongest possible guarantee, with $\ell_2/\ell_2$ coming second, $\ell_2/\ell_1$ third and exactly $k$-sparse being the weaker. We also note that all \cite{rv08,cgv13,b14,hr16} obtain improved analyses of the Restricted Isometry property; the algorithm is suggested and analyzed (modulo the RIP property) in \cite{bd08}. 
  The work in \cite{hikp12a} does not explicitly state the extension to the $d$-dimensional case, 
  but can easily be inferred from the arguments. \cite{hikp12a,ik14,k16,kvz19} work when the universe size in each dimension are powers of $2$. 
  We also assume that the signal-to-noise ratio is bounded by a polynomial of $n$, which is a standard assumption in the sparse Fourier transform literature \cite{hikp12a,ik14,k16,k17,ln19}.
} \label{tbl:table}
\end{table}

\paragraph{Comparison between $\ell_\infty/\ell_2$ and $\ell_2/\ell_2$ (or $\ell_2/\ell_1$).}
For the sake of argument, we will consider only the $\ell_2/\ell_2$ guarantee which is stronger than $\ell_2/\ell_1$.
The $\ell_2/\ell_2$ guarantee is the following: for $ \wh{x} \in \mathbb{C}^n$ one should output a $z$ such that $\|\wh{x} - z\|_2 \leq C  \| \wh{x}_{-k} \|_2$, where $C>1$ is the approximation factor. The $\ell_2/\ell_2$ guarantee can be immediatelly obtained by $\ell_\infty/\ell_2$ guarantee by truncating $z$ to its top $k$ coordinates. 
Consider $C=1.1$ \footnote{This is the case with the \textsc{RIP} based approaches, which obtain $\ell_2/\ell_1$. 
In fact, many filter-based algorithms facilitate $(1+\epsilon)$ on the right-hand side, with the number of measurements being multiplied by $\epsilon^{-1}$. 
By enabling the same dependence on $\epsilon^{-1}$ our algorithm facilitates a multiplicative $\epsilon$ factor on the right-hand side of the $\ell_\infty/\ell_2$, which makes it much stronger. Thus, a similar argument can go through.}, and think of the following signal: for a set $S$ of size $0.05k$ we have $|\wh{x}_i| =\frac{2}{\sqrt{k}} \|\wh{x}_{\bar{S}}\|_2$. 
Then the all zeros vectors is a valid solution for the $\ell_2/\ell_2$ guarantee, since 
\begin{align*}
\|\vec 0 - \wh{x}\|_2^2  = \|\wh{x}_S\|_2^2 + \|\wh{x}_{\bar{S}}\|_2^2 =  0.05 k \cdot \frac{4}{ k} \|\wh{x}_{\bar{S}}\|_2^2 + \|\wh{x}_{\bar{S}}\|_2^2 = 1.2\| \wh{x}_{\bar{S}}\|_2^2 < 1.1^2 \|\wh{x}_{\bar{S}}\|_2^2.
\end{align*}

It is clear that since $\vec 0$ is a possible output, we may not recover any of the coordinates in $S$, which is the set of ``interesting'' coordinates. 
On the other hand, the $\ell_{\infty}/\ell_2$ guarantee does allow the recovery of every coordinate in $S$. 
This is a difference between recovering all $0.05k$ versus $0$ coordinates. 
From the above discussion, one can conclude in the case where there is too much noise, $\ell_2/\ell_2$ becomes much weaker than $\ell_\infty/\ell_2$, and can be even meaningless. 
Thus, $\ell_\infty/\ell_2$ is highly desirable, whenever it is possible. 
The same exact argument holds for $\ell_2/\ell_1$.

\begin{remark}
We note that \cite{ct06,rv08,cgv13,b14,hr16} obtain a uniform guarantee, i.e. with $1-1/\mathrm{poly}(n)$ they allow reconstruction of all vectors; $\ell_\infty/\ell_2$ and $\ell_2/\ell_2$ are impossible in the uniform case, see \cite{cdd09}. We note that our comparison between the guarantees is in terms of the quality of approximation. With respect to that, $\ell_\infty/\ell_2$ is the strongest one. 
\end{remark}

\subsection{Summary of previous Filter function based technique}



One of the two ways to perform Fourier sparse recovery is by trying to implement arbitrary linear measurements, with algorithms similar to the ubiquitous \textsc{CountSketch} \cite{ccf02}. 
In the general setting \textsc{CountSketch} hashes every coordinate to one of the $O(k)$ buckets, and repeats $O( \log n)$ times with fresh randomness. 
Then, it is guaranteed that every heavy coordinate will be isolated, and the contribution from non-heavy elements is small. 
To implement this in the Fourier setting becomes a highly non-trivial task however: one gets access only to the time-domain but not the frequency domain. 
One natural way to do this is to exploit the convolution theorem and find a function which is sparse in the time domain and approximates the indicator of an interval (rectangular pulse) in the frequency domain; these functions are called (bandpass) filters. 
Appropriate filters were designed in \cite{hikp12a,hikp12b}: they were very good approximations of the rectangular pulse, i.e. the contribution from elements outside the passband zone contributed only by $1/\poly( n )$ their mass. 
These filters had an additional $\log n$ factor (in one dimension) in the sparsity of the time domain and they are sufficient for the purposes of \cite{hikp12a}, but in high dimensions this factor becomes $\log^d n$. 
Filters based on the Dirichlet kernel give a better dependence in terms of sparsity and dimension (although still an exponential dependence on the latter), but the leak to subsequent buckets, i.e. coordinates outside the passband zone contribute a constant fraction of their mass, in contrast to the filter used in \cite{hikp12a}. 
Thus one should perform additional denoising, which is a non-trivial task. 
The seminal work of Indyk and Kapralov \cite{ik14} was the first that showed how to perform sparse recovery with these filters, and then Kapralov \cite{k16,k17} extended this result to run in sublinear time. 
Note that any filter-based approach with filters which approximate the $\ell_\infty$ box suffers from the curse of dimensionality. 
\cite{kvz19} devised an algorithm which avoids the curse of dimensionality by using careful aliasing, but it works in the noiseless case and has a cubic dependence on $k$. 

\subsection{RIP property-based algorithms: a quick overview}

We say the matrix $A \in \C^{m \times n}$ satisfies RIP (Restricted Isometry Property \cite{ct05}) of order $k$ if for all $k$-sparse vectors $x \in \C^n$ we have $\|A x\|_2^2 \approx \|x \|_2^2$. 
A celebrated result of Candes and Tao \cite{ct06} shows that Basis Pursuit ($\ell_1$ minimization) suffices for sparse recovery, as long as the samples from the time domain satisfy RIP. 

In \cite{ct06} it was also proved using generic chaining that random sampling with oversampling factor $O(\log^6 n)$ gives RIP property for any orthonormal matrix with bounded entries by $1/\sqrt{n}$. 
Then \cite{rv08} improved the bound to $O(k \cdot \log^2 k \cdot \log (k \log n) \cdot \log n)$ and \cite{cgv13} improved it to $O(k \cdot \log^3k \cdot \log n)$. 
Subsequent improvement by Bourgain \cite{b14} has lead to $O(k \log k \cdot \log^2 n)$ samples, improved by Haviv and Regev to $O(k \log^2 k \cdot \log n)$\cite{hr16}. 
The fastest set of algorithms are iterative ones: for example Iterative Hard Thresholding \cite{bd09} or CoSaMP \cite{nt09} run in $O(\log n)$ iterations\footnote{To be precise, their running time is logarithmic in the signal-to-noise ratio.} and each iteration takes $\wt{O}(n)$ time.

We note the recent lower bound of \cite{r19}: a subsampled Fourier matrix that satisfies the RIP properties should have $\Omega(k \log k \cdot d)$ rows\footnote{\cite{bllm19} independently gives a similar bound for $d=\log n$.}. 
This bound is particularly useful in high dimensions, since it deteriorates to a trivial bound in low dimensions. 
We still believe though that a bound of $\Omega( k \log k \log n)$ should hold in all dimensions. 
Thus, what remains is to obtain the $\ell_2/\ell_2$ guarantee by giving a tighter analysis, and removing the one $\log k$ factor to match the lower bound, but our algorithm already allows Fourier sparse recovery with these number of samples, even with a stronger guarantee.

\subsection{Overview of our technique}

Let $x \in \C^{[p]^d}$ denote our input signal in the time domain. 
In the following we assume the knowledge of $\mu = \frac{1}{\sqrt{k}}\|\wh{x}_{-k}\|_2$ and $R^*$ which is an upper bound of $\|\wh{x}\|_\infty / \mu$, and bounded by $\poly( n )$.
These are standard assumption \cite{hikp12a,ik14,k16,k17,ln19} in the sparse Fourier transform literature. 
The bound on $R^*$ is useful for bounding the running time (or the number of measurements in \cite{hikp12a}) and in any of \cite{hikp12a,ik14,k16,k17,ln19} a $\log n$ can be substituted by $\log R^*$ in the general case, which is also the case for our algorithm. 
We note that our algorithm will be correct with probability $1-1/\poly(n)$ whenever $R^* < 2^{n^{100}}$; this is fine for every reasonable application. It might seem counter-intuitive that we need this upper bound on $R^*$, since intuitively larger signal to noise ratio should only help. However, this is an artifact of the techniques of Sparse Fourier Transform in general, either they are iterative or not.
We assumed the rounding errors in FFT computation to be negligible, similarly to Remark 3.4 in \cite{ik14}.

\subsubsection{Estimators and random shifts}

Consider the simplest scenario: $d=1$, $p$ is a prime number and a $1$-sparse signal $\wh{x}$ which is $1$ on some frequency $f^*$. 
From a sample $x_t$ in the time-domain what would be the most reasonable way to find $f^*$? For every $f \in [p]$ we would compute 
\begin{align*}
\sqrt{n} \omega^{ft}x_t = \sqrt{n} \omega^{ft} \cdot \frac{1}{\sqrt{n}} \sum_{f'\in [p]} \omega^{-f't} \wh{x}_{f'} = \omega^{(f-f^*)t},
\end{align*}
and keep, for $ t\neq 0$, the frequency that gives a real number. 
Since $(f-f^*)t$ will be zero only for $ f= f^*$, we are guaranteed correct recovery. In the noisy and multi-dimensional case or $p$ is an arbitrary integer, however, this argument will not work, because of the presence of contribution from other elements and the fact that $(f-f^*)^\top t$ can be zero modulo $p$ for other frequencies apart from $f$. 
However, we can take a number of samples $t$ and average $\sqrt{n}\omega^{f^\top t}$, and hope that this will make the contribution from other frequencies small enough, so that we can infer whether $f$ corresponds to a heavy coordinate or not. 
More specifically, we pick a set $T$ of frequencies uniformly at random from $[p]^d$ and compute
\begin{align*}
\frac{\sqrt{n}}{|T|}\sum_{t \in T} \omega^{f^\top t }x_t 
\end{align*}
for all frequencies $f$. 
We show that if $|T| = O(k)$ our estimator is good on average (and later we will maintain $O(\log n)$ independent instances and take the median to make sure with probability $1 - 1 / \poly(n)$ the estimators for all the frequencies are good), and in fact behaves like a crude filter, similarly to the ones used in \cite{ik14}, in the sense that every coordinate contributes a non-trivial amount to every other coordinate. 
However, these estimators do not suffer from the curse of dimensionality and our case is a little bit different, requiring a quite different handling. 
The main reason is that in contrast to the filters used in \cite{ik14}, there is not an easy way to formulate an isolation argument from heavy elements that would allow easy measurement re-use, like Definition 5.2 and Lemma 5.4 from \cite{ik14}. 
Buckets induced by filter functions have a property of locality, since they correspond to approximate $\ell_\infty$ boxes (with a polynomial decay outside of the box) in $[p]^d$: the closer two buckets are the more contribute the elements of one into the other. 
Our estimators on the other side do not enjoy such a property. 
Thus, one has to proceed via a different argument. 

In what follows, we will discuss how to combine the above estimators with an iterative loop that performs denoising, i.e. removes the contribution of every heavy element to other heavy elements.

\paragraph{Random shifts.} Our approach for performing denoising is quite general, and is clean on a high-level. Let $S$ be the set of the large coordinates of $\wh{x}$, i.e. those with magnitude at least $(1/\sqrt{k}) \|\wh{x}_{-k}\|_2$. We are going to estimate $\wh{x}_f$ for $f \in [p]^d$ using the estimators introduced in the previous paragraphs. Then, for those frequencies $f$ for which the values obtained are sufficiently large (larger than $\|x\|_\infty\cdot 2^{-\ell}$) , we are going to implicitly subtract that value from $\wh{x}_f$; this corresponds to updating the signal, a common trick in the Sparse Fourier literature. Then we shall iterate, re-using the same samples again till the signal to noise ratio becomes small enough. It can be shown that only coordinates in $S$ are ever updated. The trick for sample reuse in our case is the following: in the $\ell$th iteration we approximate $\mathbb{C}^{[p]^d}$ by an appropriate grid of side length $\beta\|x\|_\infty\cdot 2^{-\ell}$, where $\beta$ is an absolute constant, and then keep $O(\log n)$ random shifts of it.  Keeping $O(\log n)$ randomly shifted grids one can show that in every iteration a nice property is satisfied: we can guarantee that there exists a grid such that for all $f \in S$, $\wh{x}_f$ and its estimator round to the same point. Projecting onto that grid, what our algorithm shows is that the signal under update follows a predictable trajectory, something that allows us to argue about the correctness of the algorithm without taking an intractable union-bound over all possible trajectories the algorithm could evolve. In essence, our algorithm shows that we can at $\ell$th step compute every $\wh{x}_f$ up to $\|\wh{x}\|_\infty 2^{-\ell}$ error.  

\paragraph{A high-level explanation.}
Let us try to combine the previous two ideas. Assume that at iteration $\ell$ we have an approximation of $\wh{x}_f$ for all $f \in S$ up to $\|x\|_{\infty}\cdot 2^{-\ell}$. If we were to pick $O(k \log n \log k)$ fresh samples, we could, using our estimators, approximate $\wh{x}_f$ up to $(1/k) \|x\|_{\infty}\cdot 2^{-\ell}$. Let that approximation be $y$. We then round $y$ to $O(\log n)$ randomly shifted grids of diameter $ \|x\|_{\infty}\cdot 2^{-\ell}$. A probabilistic argument shows that, due to our choice of parameters, with high probability there exists a grid such that $\wh{x}_f$ and $y_f$ are rounded to the same grid point (the additional $O(\log k)$ factor in the sample complexity is what makes this argument go through); and we can also decide which grid this is! Thus, we safely project $y$ and be sure that what we now have at our hands is $\wh{x}$ projected onto that grid. Thus, in the next iteration $\ell+1$ we only need to argue correctness for at most $O(\log n)$ vectors, that is, vectors $\wh{x}$ rounded on one of the aforementioned $O(\log n)$ grids. This dramatically decreases the number of events we have to analyze, and we can set up an inductive argument that guarantees the correctness of the algorithm. Note that there is no independence issue, since the randomness between the samples taken (used for the estimators) and the projection onto the randomly shifted grids is not shared.\newline

We first implement a procedure which takes $O(k \log n)$ uniform random measurements from $x$ and has the guarantee that for any $\nu \geq \mu$ any $y \in \C^{[p]^d}$ where $\|\wh{x} - y \|_\infty \leq 2 \nu$ and $y$ is independent from the randomness of the measurements, the procedure outputs a $O(k)$-sparse $z \in \C^{[p]^d}$ such that $\|\wh{x} - y - z\|_\infty \leq \nu$ with probability $1 - 1 / \poly(n)$. 

\begin{lemma} [\textsc{LinfinityReduce} procedure, informal]\label{lem:locate_and_estimate_informal}
Let $\mu = \frac{1}{\sqrt{k}} \|\wh{x}_{-k}\|_2$, and $\nu \geq \mu$.
Let ${\cal T}^{(0)}$ be a list of $O(k \log n)$ i.i.d. elements in $[p]^d$. 
Let $S$ be top $O(k)$ coordinates in $\wh{x}$. 
There is a procedure that takes $\{x_t\}_{t \in {\cal T}}$, $y \in \C^{[p]^d}$ and $\nu$ as input, runs in $\wt{O}(n)$ time, and outputs $z \in \C^{[p]^d}$ so that if $\| \wh{x} - y \|_\infty \leq 2 \nu$, 
$\supp(y) \subseteq S$ and $y$ is independent from the randomness of ${\cal T}^{(0)}$, then $\| \wh{x} - y - z \|_\infty \leq \nu$ and $\supp(z) \subseteq S$ with probability $1 - 1 / \poly(n)$ under the randomness of ${\cal T}^{(0)}$.
\end{lemma}

Namely, we can take $O(k \log n)$ measurements and run the procedure in Lemma~\ref{lem:locate_and_estimate_informal} to reduce (the upper bound of) the $\ell_\infty$ norm of the residual signal by half.
We call the procedure in Lemma~\ref{lem:locate_and_estimate_informal} \textsc{LinfinityReduce} procedure.
More generally, we can take $O(H \cdot k \log n)$ measurements and run the \textsc{LinfinityReduce} procedure $H$ times to reduce the $\ell_\infty$ norm of the residual signal to $1 / 2^H$ of its original magnitude, with failure probability at most $1 / \poly(n)$.
Note that if we set $H = \log R^*$, we have already obtained a 
m taking $O(k \log n \log R^*)$ measurements, because we can drive down (the upper bound of) the $\ell_\infty$ norm of the residual signal from $\|\wh{x}\|_\infty$ to $\mu$ in $\log R^*$ iterations. 

\subsubsection{$O(k \log n)$ samples for $k = O(\log n)$}\label{sec:warmup}
We first discuss a measurement reuse idea that leads us to a sparse recovery algorithm (Algorithm~\ref{alg:fourier_sparse_recovery_k_logn}) taking $O(k \log n)$ measurements for $k = O(\log n)$.
We set $H = 5$, and let ${\cal T} = \{{\cal T}^{(1)}, \ldots, {\cal T}^{(H)}\}$, where each ${\cal T}^{(h)}$ is a list of $O(k \log n)$ i.i.d. elements in $[p]^d$.
Note that ${\cal T}^{(1)}, \ldots, {\cal T}^{(H)}$ are independent.
In our sparse Fourier recovery algorithm, we will measure $x_t$ for all $t \in {\cal T}$.

In a nutshell, our approach finely discretizes the space of possible trajectories the algorithm could evolve, and carefully argues about the correctness of the algorithm by avoiding the intractable union-bound over all trajectories.

\paragraph{Recovery algorithm.} The recovery algorithm proceeds in $\log R^* - H + 1$ iterations, where each iteration (except the last iteration) the goal is to reduce the upper bound of $\ell_\infty$ norm of the residual signal by half. 
Initially, the upper bound is $R^*$.
It is important to note that we use the same measurements ${\cal T} = \{{\cal T}^{(1)}, \ldots, {\cal T}^{(H)}\}$ in all of these $\log R^* - H + 1$ iterations.

In the following, we will describe one iteration of the recovery algorithm. 
Let $y \in \C^{[p]^d}$ denote the sparse vector recovered so far, and let the upper bound of $\|\wh{x} - y\|_\infty$ be $2 \nu$. 
Running the \textsc{LinfinityReduce} procedure $H$ times where in the $h$-th time we use measurements in ${\cal T}^{(h)}$, we obtain a $O(k)$-sparse $z$ such that with probability $1 - 1 / \poly(n)$, $\|\wh{x} - y - z\|_\infty \leq 2^{1-H} \nu \leq 0.1 \nu$ (we call such $z$ a desirable output by the \textsc{LinfinityReduce} procedure).
Instead of taking $y+z$ as our newly recovered sparse signal, for each $f \in \supp(y+z)$, we project $y_f+z_f$ to the nearst points in $\mathcal{G}_{0.6 \nu}:= \{0.6 \nu (x + y \i): x,y \in \Z \}$ and assign to $y'_f$, where $y'$ denotes our newly recovered sparse signal.
For all $f \not\in \supp(y+z)$, we let $y'_f = 0$.

To simplify our exposition, here we introduce some notations.
We call $\mathcal{G}_{0.6 \nu}$ a grid of side length $0.6 \nu$, and we generalize the definition to any side length.
Namely, for any $r_g > 0$, let grid $\mathcal{G}_{r_g}:= \{r_g(x + y \i): x,y \in \Z \}$.
Moreover, we define $\Pi_{r_g}: \C \rightarrow \mathcal{G}_{r_g}$ to be the mapping that maps any element in $\C$ to the nearest element in $\mathcal{G}_{r_g}$.
Now we can write $y'$ as
\begin{align*}
y'_f = 
\begin{cases} 
  \Pi_{0.6 \nu}(y_f + z_f), & \text{~if~} f \in \supp(y+z); \\ 
  0, & \text{~if~} f \not\in \supp(y+z).
\end{cases} 
\end{align*}

At the end of each iteration, we assign $y'$ to $y$, and shrink $\nu$ by half. In the last iteration, we will not compute $y'$, instead we output $y+z$. 
We present the algorithm in Algorithm~\ref{alg:fourier_sparse_recovery_k_logn}.

\begin{algorithm}[!t]\caption{Fourier sparse recovery by projection, $O(k \log n)$ measurements when $k = O(\log n)$}\label{alg:fourier_sparse_recovery_k_logn}
\begin{algorithmic}[1]
\Procedure{\textsc{FourierSparseRecoveryByProjection}}{$x,n,k,\mu,R^*$} \Comment{Section~\ref{sec:warmup}} 
  \State {\bf Require} that $\mu = \frac{1}{\sqrt{k}} \|\wh{x}_{-k}\|_2$ and $R^* \geq \|\wh{x}\|_\infty\ / \mu$ 
  \State $H \leftarrow 5$, $\nu \leftarrow \mu R^* / 2$, ~$y \leftarrow \vec 0$ \Comment{$y \in \C^{[p]^d}$ refers to the sparse vector recovered so far}
  \State Let ${\cal T} = \{ {\cal T}^{(1)}, \cdots, {\cal T}^{(H)} \}$ where each ${\cal T}^{(h)}$ is a list of i.i.d. uniform samples in $[p]^d$ 
  \While {{\bf true}} 
    \State $\nu' \leftarrow 2^{1-H} \nu$
    \State Use $\{x_t\}_{t\in \mathcal{T}}$ to run the \textsc{LinfinityReduce} procedure (in Lemma~\ref{lem:locate_and_estimate_informal}) $H$ times (use samples in ${\cal T}^{(h)}$ for each $h\in [H]$ ), and finally it finds $z$ so that $\|\wh{x} - y - z\|_\infty \leq \nu'$ 
    \State{{\bf if} {$\nu' \leq \mu$} } {\bf then} \Return $y+z$ \Comment{We found the solution}
    \State $y' \leftarrow \vec 0$
    \For{$f \in \supp(y+z)$}
      \State $y'_f \leftarrow \Pi_{0.6 \nu}(y_f + z_f)$ \Comment{We want $\|\wh{x} - y' \|_\infty \leq \nu$ and the depend-}
    \EndFor \Comment{ence between $y'$ and $\mathcal{T}$ is under control~~~}
    \State $y \leftarrow y'$, $\nu \leftarrow \nu / 2$
  \EndWhile
\EndProcedure
\end{algorithmic}
\end{algorithm}

\paragraph{Analysis.} We analyze $y'$ conditioned on the event that $\|\wh{x} - y - z\|_\infty \leq 0.1 \nu$ (i.e. $z$ is a desirable output by the \textsc{LinfinityReduce} procedure, which happens with probability $1 - 1 / \poly(n)$).
We will prove that $y'$ has two desirable properties: (1) $\|\wh{x} - y'\|_\infty \leq \nu$; (2) the dependence between $y'$ and our measurements ${\cal T}$ is under control so that after taking $y'$ as newly recovered sparse signal, subsequent executions of the \textsc{LinfinityReduce} procedure with measurements ${\cal T}$ still work with good probability. 
Property (1) follows from triangle inequality and the fact that $\|\wh{x} - (y+z)\|_\infty \leq 0.1 \nu$ and $\|(y+z) - y'\|_\infty \leq 0.6 \nu$.
We now elaborate on property (2). 
We can prove that for any $f \in [p]^d$, 
\begin{align*}
y'_f \in \big\{\Pi_{0.6\nu}(\wh{x}_f+0.1\nu(\alpha+ \beta\i)): \alpha, \beta \in \{-1, 1\} \big\}.
\end{align*}
Let $S$ denote top $26 k$ coordinates (in absolute value) of $\wh{x}$.
We can further prove that for any $f \in \ov{S}$, $y'_f =0$.
Therefore, the total number of possible $y'$ is upper bounded by $4^{|S|} = 4^{O(k)}$.
If $k = O(\log n)$, we can afford union bounding all $4^{O(k)} = \poly(n)$ possible $y'$, and prove that with probability $1 - 1 / \poly(n)$ for all possible value of $y'$ if we take $y'$ as our newly recovered sparse signal then in the next iteration the \textsc{LinfinityReduce} procedure with measurements ${\cal T}$ gives us a desirable output.

\paragraph{Sufficient event.}
More rigorously, we formulate the event that guarantees successful execution of Algorithm~\ref{alg:fourier_sparse_recovery_k_logn}.
Let ${\cal E}_1$ be the event that for all $O(\log R^*)$ possible values of $\nu \in \{ \mu \frac{ R^*}{2}, \mu \frac{ R^*}{4}, \ldots, \mu 2^{H-1} \}$,
for all possible vector $y$ where $y_f = 0$ for $f \in \ov{S}$ and $y_f \in \{ \Pi_{0.6 \nu} (\wh{x}_f + 0.1 \nu(\alpha + \beta \i)): \alpha, \beta \in \{-1, 1\} \}$ for $f \in S$ (we also need to include the case that $y = \vec 0$ for the success of the first iteration), 
running the \textsc{LinfinityReduce} procedure (in Lemma~\ref{lem:locate_and_estimate_informal}) $H$ times (where in the $h$-th time measurements $\{x_t\}_{t \in {\cal T}^{(h)}}$ are used to reduce the error from $2^{2-h} \nu$ to $2^{1-h} \nu$) finally gives $z$ so that $\|\wh{x} - y - z\|_\infty \leq 2^{1-H} \nu$.
The randomness of ${\cal E}_1$ comes from ${\cal T} = \{ {\cal T}^{(1)}, \ldots, {\cal T}^{(H)} \}$.

First, event ${\cal E}_1$ happens with probability $1 - 1 / \poly(n)$. This is because there are $4^{O(k)} \log R^*$ possible combinations of $\nu$ and $y$ to union bound, and each has failure probability at most $1 / \poly(n)$.
For $k = O(\log n)$, and any $R^* < 2 ^{n^{100}}$ this gives the desired result.
Second, conditioned on event ${\cal E}_1$ happens, Algorithm~\ref{alg:fourier_sparse_recovery_k_logn} gives correct output.
This can be proved by a mathematical induction that in the $t$-th iteration of the while-true loop in Algorithm~\ref{alg:fourier_sparse_recovery_k_logn}, $\|\wh{x} - y\|_\infty \leq 2^{-t} \mu R^*$.

\subsubsection{$O(k \log k\log n)$ samples suffice}

We first introduce some notations.
For any $r_g > 0$, define the grid $\mathcal{G}_{r_g}:= \{r_g(x + y \i): x,y \in \Z \}$.
Moreover, we define $\Pi_{r_g}: \C \rightarrow \mathcal{G}_{r_g}$ to be the mapping that maps any element in $\C$ to the nearest element in $\mathcal{G}_{r_g}$.

\paragraph{Using random shift to reduce projection size.}


We introduce the random shift trick, the property of which is captured by Lemma~\ref{lem:random_shift_informal}. To simplify notation, for any $r_b > 0$ and $c \in \C$ we define box $\mathcal{B}_\infty(c, r_b) := \{c + r_b(x+y \i): x,y \in [-1,1]\}$.
For any $S \subseteq \C$, let $\Pi_{r_g}(S) = \{\Pi_{r_g}(c) : c \in S\}$.

\begin{lemma} [property of a randomly shifted box, informal] \label{lem:random_shift_informal} 
Take a box of side length $2r_b$ and shift it randomly by an offset in $\mathcal{B}_\infty(0, r_s)$ (or equivalently, $[-r_s, r_s] \times [-r_s, r_s]$) where $r_s \geq r_b$. Next round every point inside that shifted box to the closest point in $G_{r_g}$ where $r_g \geq 2 r_s$. Then, with probability at least $(1 - r_b/r_s)^2$ everyone will be rounded to the same point.
\end{lemma}

In the following, we present a sparse Fourier recovery algorithm that incorporates the random shift idea.
The algorithm takes $O(k \log k \log n)$ measurements. We set $H = O(\log k)$ and take measurements of ${\cal T} = \{{\cal T}^{(1)}, \ldots, {\cal T}^{(H)}\}$ , where ${\cal T}^{(h)}$ is a list of $O(k \log n)$ i.i.d elements in $[p]^d$. 

In a nutshell, our approach finely discretizes the space of possible trajectories the algorithm could evolve. After we find estimates for $\wh{x}_f$ we shift them randomly and project them onto a coarse grid (which is the same as projecting onto one of randomly shifted grids). We shall show that then the number of trajectories is pruned, and we need to argue for a much smaller collection of events.
We note that we make the decoding algorithm randomized: the randomness in previous algorithms was present only when taking samples, and the rest of the algorithm was deterministic. However, here we need randomness in both cases, and that helps us prune the number of possible trajectories. To the best of our knowledge, this is a novel argument and approach, and might be helpful for future progress in the field.

\paragraph{Recovery algorithm.} We assume that we have already obtained a $O(k)$-sparse $y \in \C^{[p]^d}$ such that $\|\wh{x} - y\|_\infty \leq 2 \nu$ and $y$ is ``almost'' independent from ${\cal T}$.
We show how to obtain $y' \in \C^{[p]^d}$ such that $\|\wh{x} - y'\|_\infty \leq \nu$ with probability $1 - 1 / \poly(n)$ and $y'$ is ``almost'' independent from ${\cal T}$.
The idea is the following. We first run the \textsc{LinfinityReduce} procedure $H = O(\log k)$ times to get an $O(k)$-sparse $z \in \C^{[p]^d}$ such that $\|\wh{x} - y - z\|_\infty \leq \frac{1}{2^{20}k} \nu$.
Then we repeatedly sample a uniform random shift $s \in [- 10^{-3} \nu,  10^{-3} \nu] + \i  [-10^{-3} \nu,  10^{-3} \nu]$ until for every $f \in \supp(y+z)$, all the points (or complex numbers) of the form $y_f + z_f +s + a + b \i$ with $a,b \in [-\frac{\nu}{2^{20}k}, \frac{\nu}{2^{20}k}]$ round to the same grid point in $\mathcal{G}_{0.04 \nu}$.
Finally, for every $f \in \supp(y+z)$, we assign $\Pi_{0.04 \nu} (y_f + z_f + s)$ to $y'_f$; all remaining coordinates in $y'$ will be assigned $0$.
We present an informal version of our algorithm in Algorithm~\ref{alg:fourier_sparse_recovery_informal}, and defer its formal version to the appendix.

\begin{algorithm}[!t]\caption{Fourier sparse recovery by random shift and projection (informal version)}\label{alg:fourier_sparse_recovery_informal} 
\begin{algorithmic}[1]
\Procedure{\textsc{FourierSparseRecovery}}{$x,n,k,\mu,R^*$} \Comment{Theorem~\ref{thm:fourier_sparse_recovery_informal}, $n=p^d$} 
  \State {\bf Require} that $\mu = \frac{1}{\sqrt{k}} \|\wh{x}_{-k}\|_2$ and $R^* \geq \|\wh{x}\|_\infty\ / \mu$
  \State $H \leftarrow O(\log k)$, $\nu \leftarrow \mu R^* / 2$, ~$y \leftarrow \vec 0$ \Comment{$y \in \C^{[p]^d}$ refers to the sparse vector recovered so far}
  \State  Let ${\cal T} = \{ {\cal T}^{(1)}, \cdots, {\cal T}^{(H)} \}$ where each ${\cal T}^{(h)}$ is a list of i.i.d. uniform samples in $[p]^d$ 
  \While {{\bf true}} 
    \State $\nu' \leftarrow \frac{1}{2^{20}k} \nu$
    \State Use $\{x_t\}_{t\in \mathcal{T}}$ to run the \textsc{LinfinityReduce} procedure (in Lemma~\ref{lem:locate_and_estimate_informal}) $H$ times (use samples in ${\cal T}^h$ for each $h\in [H]$ ), and finally it finds $z$ so that $\|\wh{x} - y - z\|_\infty \leq \nu'$ 
    \State{{\bf if} {$\nu' \leq \mu$} } {\bf then} \Return $y+z$ \Comment{We found the solution}
    \Repeat 
      \State Pick $s \in \mathcal{B}_\infty(0, 10^{-3}\nu)$ uniformly at random \label{lin:s_informal}
    \Until{$\forall f \in \supp(y+z), |\Pi_{0.04 \nu}(\mathcal{B}_\infty(y_f + z_f + s, \nu'))|=1$} \label{lin:until_good_informal}
    \State $y' \leftarrow \vec 0$	\label{lin:start_project}
    \For{$f \in \supp(y+z)$}
      \State $y'_f \leftarrow \Pi_{0.04 \nu}(y_f + z_f + s)$ \Comment{We want $\|\wh{x} - y' \|_\infty \leq \nu$ and the depend-}
    \EndFor \Comment{ence between $y'$ and $\mathcal{T}$ is under control~~~} 
    \State $y \leftarrow y'$, $\nu \leftarrow \nu / 2$ \label{lin:end_project}
  \EndWhile
\EndProcedure
\end{algorithmic}
\end{algorithm}

\paragraph{Informal Analysis.} We analyze the above approach.

At every iteration, our algorithm holds a vector $y$, and computes a vector $z$. Instead of setting $y$ to $y+z$ and iterating, as would be the natural thing to do, we set $y$ to $y'$ (lines \ref{lin:start_project} to \ref{lin:end_project}) where
\begin{align*}
y'_f = 
\begin{cases} 
  \Pi_{0.6 \nu}(y_f + z_f), & \text{~if~} f \in \supp(y+z); \\ 
  0, & \text{~if~} f \not\in \supp(y+z).
\end{cases} 
\end{align*}

First, we have the guarantee that $\|\wh{x} - y'\|_\infty \leq \nu$.
Moreover, by our choice of $s$, for every $f \in \supp(y+z)$, $y_f + z_f + s$ and $\wh{x}_f + s$ round to the same grid point in $\mathcal{G}_{0.04 \nu}$. 
Therefore, for the new vector $y'$ we have recovered, we ``hide'' the randomness in ${\cal T}$, and the randomness only leaks from failed attempts of the shifts.
In the following, we show that each attempt of shift succeeds with probability $\frac{1}{2}$.

We can restate the procedure of choosing $s$ to be: 
\begin{align*}
\text{{\bf repeatedly}} & ~ \text{~sample~} s \sim \mathcal{B}_\infty(0, 10^{-3}\nu), \\
\text{{\bf until}} & ~ \text{~for~all~} f \in \supp(y+z), \Big|\Pi_{0.04\nu} \Big( \mathcal{B}_\infty \big( y_f + z_f +s, \frac{\nu}{2^{20}k} \big) \Big) \Big| = 1.
\end{align*}
Note that $|\supp(y+z)| = O(k)$. 
Let us say that we can always guarantee that $|\supp(y+z)| \leq 50 k$.
By Lemma~\ref{lem:random_shift_informal} where we let $r_b = \frac{\nu}{2^{20}k}$, $r_s = 10^{-3} \nu$ and $r_g = 0.04 \nu$, for $f \in \supp(y+z)$,
\begin{align*}
\Pr \Bigg[ \Big| \Pi_{0.04\nu} \left( \mathcal{B}_\infty \big(y_f + z_f +s, \frac{\nu}{2^{20}k} \big) \right) \Big| = 1 \Bigg] \geq \left(1 - \frac{r_b}{r_s}\right)^2 \geq 1 - \frac{1}{100k}.
\end{align*}
By a union bound over $f \in \supp(y + z)$, the probability is at least $\frac{1}{2}$ that for all $f \in \supp(y+z)$, $|\Pi_{0.04\nu}(\mathcal{B}_\infty(y_f + z_f +s, \frac{\nu}{2^{20}k}))| = 1$.

Therefore, with probability $1 - 1 / \poly(n)$, we will only try $O(\log n)$ shifts.
We can apply a union bound over $O(\log n)$ possible shifts, and prove that with probability $1 - 1 / \poly(n)$ if taking $y'$ as our new $y$, and shrinking $\nu$ by half, the \textsc{LinfinityReduce} procedure will work as desired as if there is no dependence issue.

\paragraph{Sufficient event.}
Let $S$ be top $O(k)$ coordinates in $\wh{x}$ which are also larger than $(1/\sqrt{k}) \|x_{-k}\|_2$ in magnitude.
Let $L = O(\log R^*)$ denote the number of iterations in Algorithm~\ref{alg:fourier_sparse_recovery_informal}.
For $\ell \in [L]$, let $\nu_\ell = 2^{-\ell} \mu R^*$.
For $\ell \in [L-1]$, let $s_\ell^{(a)}$ be the $a$-th sample from $\mathcal{B}_\infty(0, 10^{-3} \nu_\ell)$ as appeared on Line~\ref{lin:s_informal} in Algorithm~\ref{alg:fourier_sparse_recovery_informal}.
For the sake of analysis, we assume that Algorithm~\ref{alg:fourier_sparse_recovery_informal} actually produces an infinite sequence of shifts $s_\ell^{(1)}, s_\ell^{(2)}, \ldots$.
We formulate the event that guarantees successful execution of Algorithm~\ref{alg:fourier_sparse_recovery_informal}.
We define event $\mathcal{E}_2$ to be the union of all the following events. \\
1. For all $\ell \in [L-1]$, there exists $a \in [10 \log n]$ so that for all $f \in S$,
\begin{align*}
\left| \Pi_{0.04 \nu_\ell} \left(\mathcal{B}_\infty(\wh{x}_f + s_\ell^{(a)}, \frac{1}{100k} \nu_\ell) \right) \right| = 1.
\end{align*} 
\\
2. For $\ell = 1$, if we run the \textsc{LinfinityReduce} procedure $H$ times with $y = \vec 0$ and measurements in ${\cal T}$, we get $z$ such that $\|\wh{x} - z\|_\infty \leq 2^{1-H} \nu_1$ and $\supp(z) \subseteq S$. \\
3. For all $\ell \in \{2, \ldots, L\}$, for all $a \in [10 \log n]$, if we run the \textsc{LinfinityReduce} procedure $H$ times with $y = \xi$ where
\begin{align*}
\xi_f = 
\begin{cases} 
  \Pi_{0.04 \nu_\ell}(\wh{x}_f + s_{\ell-1}^{(a)}), & \text{~if~} f \in S; \\ 
  0, & \text{~if~} f \in \ov{S}.
\end{cases} 
\end{align*}
then we get $z$ such that $\|\wh{x} - y - z\|_\infty \leq 2^{1-H} \nu_\ell$  and $\supp(y+z) \subseteq S$.

We can prove that event ${\cal E}_2$ happens with probability $1 - 1 / \poly(n)$.
Moreover, we can prove that conditioned on event ${\cal E}_2$ Algorithm~\ref{alg:fourier_sparse_recovery_informal} gives correct output.
We defer both proofs in the appendix.







\newpage

\appendix

\newpage
\addcontentsline{toc}{section}{References}
\bibliographystyle{alpha}
\bibliography{ref}

\newcommand{\etalchar}[1]{$^{#1}$}
\begin{thebibliography}{AKM{\etalchar{+}}19}

\bibitem[AGS03]{ags03}
Adi Akavia, Shafi Goldwasser, and Shmuel Safra.
\newblock Proving hard-core predicates using list decoding.
\newblock In {\em FOCS}, volume~44, pages 146--159, 2003.

\bibitem[AKM{\etalchar{+}}19]{hkmmvz19}
Haim Avron, Michael Kapralov, Cameron Musco, Christopher Musco, Ameya
  Velingker, and Amir Zandieh.
\newblock A universal sampling method for reconstructing signals with simple
  fourier transforms.
\newblock In {\em STOC}. arXiv preprint arXiv:1812.08723, 2019.

\bibitem[ASSN08]{assn08}
Abiodun~M Aibinu, Momoh~JE Salami, Amir~A Shafie, and Athaur~Rahman Najeeb.
\newblock {MRI} reconstruction using discrete {F}ourier transform: a tutorial.
\newblock {\em World Academy of Science, Engineering and Technology}, 42:179,
  2008.

\bibitem[BCG{\etalchar{+}}14]{bcgls14}
Petros Boufounos, Volkan Cevher, Anna~C Gilbert, Yi~Li, and Martin~J Strauss.
\newblock What's the frequency, {K}enneth?: Sublinear {F}ourier sampling off
  the grid.
\newblock In {\em Algorithmica(A preliminary version of this paper appeared in
  the Proceedings of RANDOM/APPROX 2012, LNCS 7408, pp.61--72)}, pages 1--28.
  Springer, 2014.

\bibitem[BD08]{bd08}
Thomas Blumensath and Mike~E Davies.
\newblock Iterative thresholding for sparse approximations.
\newblock {\em Journal of Fourier analysis and Applications}, 14(5-6):629--654,
  2008.

\bibitem[BD09a]{bd09}
Thomas Blumensath and Mike~E Davies.
\newblock Iterative hard thresholding for compressed sensing.
\newblock {\em Applied and computational harmonic analysis}, 27(3):265--274,
  2009.

\bibitem[BD09b]{bd09short}
Thomas Blumensath and Mike~E Davies.
\newblock A simple, efficient and near optimal algorithm for compressed
  sensing.
\newblock {\em 2009 IEEE International Conference on Acoustics, Speech and
  Signal Processing}, 2009.

\bibitem[BD10]{bd10}
Thomas Blumensath and Mike~E Davies.
\newblock Normalized iterative hard thresholding: Guaranteed stability and
  performance.
\newblock {\em IEEE Journal of selected topics in signal processing},
  4(2):298--309, 2010.

\bibitem[BLLM19]{bllm19}
Jaroslaw Blasiok, Patrick Lopatto, Kyle Luh, and Jake Marcinek.
\newblock Sparse reconstruction from hadamard matrices: A lower bound.
\newblock In {\em 60th Annual IEEE Symposium on Foundations of Computer Science
  (FOCS)}. \url{https://arxiv.org/pdf/1903.12135.pdf}, 2019.

\bibitem[Bou14]{b14}
Jean Bourgain.
\newblock An improved estimate in the restricted isometry problem.
\newblock In {\em Geometric Aspects of Functional Analysis}, pages 65--70.
  Springer, 2014.

\bibitem[BZI17]{bzi17}
Sina Bittens, Ruochuan Zhang, and Mark~A Iwen.
\newblock A deterministic sparse {FFT} for functions with structured {F}ourier
  sparsity.
\newblock {\em Advances in Computational Mathematics, to appear.}, 2017.

\bibitem[CCF02]{ccf02}
Moses Charikar, Kevin Chen, and Martin {Farach-Colton}.
\newblock Finding frequent items in data streams.
\newblock In {\em Automata, Languages and Programming}, pages 693--703.
  Springer, 2002.

\bibitem[CDD09]{cdd09}
Albert Cohen, Wolfgang Dahmen, and Ronald DeVore.
\newblock Compressed sensing and best $k$-term approximation.
\newblock {\em Journal of the American mathematical society}, 22(1):211--231,
  2009.

\bibitem[CGV13]{cgv13}
Mahdi Cheraghchi, Venkatesan Guruswami, and Ameya Velingker.
\newblock Restricted isometry of {F}ourier matrices and list decodability of
  random linear codes.
\newblock {\em SIAM Journal on Computing}, 42(5):1888--1914, 2013.

\bibitem[CI17]{ci17}
Mahdi Cheraghchi and Piotr Indyk.
\newblock Nearly optimal deterministic algorithm for sparse walsh-hadamard
  transform.
\newblock {\em ACM Transactions on Algorithms (TALG)}, 13(3):34, 2017.

\bibitem[CKPS16]{ckps16}
Xue Chen, Daniel~M Kane, Eric Price, and Zhao Song.
\newblock {F}ourier-sparse interpolation without a frequency gap.
\newblock In {\em Foundations of Computer Science (FOCS), 2016 IEEE 57th Annual
  Symposium on}, pages 741--750. IEEE, 2016.

\bibitem[CP19a]{cp19b}
Xue Chen and Eric Price.
\newblock Active regression via linear-sample sparsification.
\newblock In {\em Conference on Learning Theory, {COLT} 2019, 25-28 June 2019,
  Phoenix, AZ, {USA}}, pages 663--695, 2019.

\bibitem[CP19b]{cp19a}
Xue Chen and Eric Price.
\newblock Estimating the frequency of a clustered signal.
\newblock In {\em 46th International Colloquium on Automata, Languages, and
  Programming, {ICALP} 2019, July 9-12, 2019, Patras, Greece.}, pages
  36:1--36:13, 2019.

\bibitem[CRT06a]{crt06a}
Emmanuel~J Cand{\`e}s, Justin Romberg, and Terence Tao.
\newblock Robust uncertainty principles: Exact signal reconstruction from
  highly incomplete frequency information.
\newblock {\em IEEE Transactions on information theory}, 52(2):489--509, 2006.

\bibitem[CRT06b]{crt06}
Emmanuel~J Candes, Justin~K Romberg, and Terence Tao.
\newblock Stable signal recovery from incomplete and inaccurate measurements.
\newblock {\em Communications on pure and applied mathematics},
  59(8):1207--1223, 2006.

\bibitem[CT65]{ct65}
James~W Cooley and John~W Tukey.
\newblock An algorithm for the machine calculation of complex {F}ourier series.
\newblock {\em Mathematics of computation}, 19(90):297--301, 1965.

\bibitem[CT05]{ct05}
Emmanuel Candes and Terence Tao.
\newblock Decoding by linear programming.
\newblock {\em arXiv preprint math/0502327}, 2005.

\bibitem[CT06]{ct06}
Emmanuel~J Candes and Terence Tao.
\newblock Near-optimal signal recovery from random projections: Universal
  encoding strategies?
\newblock {\em IEEE transactions on information theory}, 52(12):5406--5425,
  2006.

\bibitem[DDTS06]{ddts06}
David~Leigh Donoho, Iddo Drori, Yaakov Tsaig, and Jean-Luc Starck.
\newblock {\em Sparse solution of underdetermined linear equations by stagewise
  orthogonal matching pursuit}.
\newblock Department of Statistics, Stanford University, 2006.

\bibitem[DM08]{dm08}
Wei Dai and Olgica Milenkovic.
\newblock Subspace pursuit for compressive sensing: Closing the gap between
  performance and complexity.
\newblock Technical report, ILLINOIS UNIV AT URBANA-CHAMAPAIGN, 2008.

\bibitem[Don06]{d06}
David~L. Donoho.
\newblock Compressed sensing.
\newblock {\em {IEEE} Trans. Information Theory}, 52(4):1289--1306, 2006.

\bibitem[Fou11]{f11}
Simon Foucart.
\newblock Hard thresholding pursuit: an algorithm for compressive sensing.
\newblock {\em SIAM Journal on Numerical Analysis}, 49(6):2543--2563, 2011.

\bibitem[GGI{\etalchar{+}}02]{ggims02}
Anna~C Gilbert, Sudipto Guha, Piotr Indyk, S~Muthukrishnan, and Martin Strauss.
\newblock Near-optimal sparse {F}ourier representations via sampling.
\newblock In {\em Proceedings of the thiry-fourth annual ACM symposium on
  Theory of computing}, pages 152--161. ACM, 2002.

\bibitem[GK09]{gk09}
Rahul Garg and Rohit Khandekar.
\newblock Gradient descent with sparsification: an iterative algorithm for
  sparse recovery with restricted isometry property.
\newblock In {\em ICML}, volume~9, pages 337--344, 2009.

\bibitem[GL89]{gl89}
Oded Goldreich and Leonid~A Levin.
\newblock A hard-core predicate for all one-way functions.
\newblock In {\em Proceedings of the twenty-first annual ACM symposium on
  Theory of computing}, pages 25--32. ACM, 1989.

\bibitem[GMS05]{gms05}
Anna~C Gilbert, S~Muthukrishnan, and Martin Strauss.
\newblock Improved time bounds for near-optimal sparse {F}ourier
  representations.
\newblock In {\em Optics \& Photonics 2005}, pages 59141A--59141A.
  International Society for Optics and Photonics, 2005.

\bibitem[Goo05]{g05}
Joseph~W Goodman.
\newblock {\em Introduction to {F}ourier optics}.
\newblock Roberts and Company Publishers, 2005.

\bibitem[HIKP12a]{hikp12a}
Haitham Hassanieh, Piotr Indyk, Dina Katabi, and Eric Price.
\newblock Nearly optimal sparse {F}ourier transform.
\newblock In {\em Proceedings of the forty-fourth annual ACM symposium on
  Theory of computing}, pages 563--578. ACM, 2012.

\bibitem[HIKP12b]{hikp12b}
Haitham Hassanieh, Piotr Indyk, Dina Katabi, and Eric Price.
\newblock Simple and practical algorithm for sparse {F}ourier transform.
\newblock In {\em Proceedings of the twenty-third annual ACM-SIAM symposium on
  Discrete Algorithms}, pages 1183--1194. SIAM, 2012.

\bibitem[HR16]{hr16}
Ishay Haviv and Oded Regev.
\newblock The restricted isometry property of subsampled {F}ourier matrices.
\newblock In {\em SODA}, pages 288--297. arXiv preprint arXiv:1507.01768, 2016.

\bibitem[IK14]{ik14}
Piotr Indyk and Michael Kapralov.
\newblock Sample-optimal {F}ourier sampling in any constant dimension.
\newblock In {\em Foundations of Computer Science (FOCS), 2014 IEEE 55th Annual
  Symposium on}, pages 514--523. IEEE, 2014.

\bibitem[IKP14]{ikp14}
Piotr Indyk, Michael Kapralov, and Eric Price.
\newblock ({N}early) {S}ample-optimal sparse {F}ourier transform.
\newblock In {\em Proceedings of the Twenty-Fifth Annual ACM-SIAM Symposium on
  Discrete Algorithms}, pages 480--499. SIAM, 2014.

\bibitem[Iwe08]{i08}
Mark~A Iwen.
\newblock A deterministic sub-linear time sparse {F}ourier algorithm via
  non-adaptive compressed sensing methods.
\newblock In {\em Proceedings of the nineteenth annual ACM-SIAM symposium on
  Discrete algorithms}, pages 20--29. Society for Industrial and Applied
  Mathematics, 2008.

\bibitem[Iwe10]{i10}
Mark~A Iwen.
\newblock Combinatorial sublinear-time {F}ourier algorithms.
\newblock {\em Foundations of Computational Mathematics}, 10(3):303--338, 2010.

\bibitem[Iwe13]{iw13}
Mark~A Iwen.
\newblock Improved approximation guarantees for sublinear-time {F}ourier
  algorithms.
\newblock {\em Applied And Computational Harmonic Analysis}, 34(1):57--82,
  2013.

\bibitem[Kap16]{k16}
Michael Kapralov.
\newblock Sparse {F}ourier transform in any constant dimension with
  nearly-optimal sample complexity in sublinear time.
\newblock In {\em Symposium on Theory of Computing Conference, STOC'16,
  Cambridge, MA, USA, June 19-21, 2016}, 2016.

\bibitem[Kap17]{k17}
Michael Kapralov.
\newblock Sample efficient estimation and recovery in sparse {FFT} via
  isolation on average.
\newblock In {\em Foundations of Computer Science (FOCS), 2017 IEEE 58th Annual
  Symposium on}, pages 651--662. Ieee, 2017.

\bibitem[KM93]{km93}
Eyal Kushilevitz and Yishay Mansour.
\newblock Learning decision trees using the {F}ourier spectrum.
\newblock {\em SIAM Journal on Computing}, 22(6):1331--1348, 1993.

\bibitem[KVZ19]{kvz19}
Michael Kapralov, Ameya Velingker, and Amir Zandieh.
\newblock Dimension-independent sparse {F}ourier transform.
\newblock In {\em Proceedings of the Thirtieth Annual ACM-SIAM Symposium on
  Discrete Algorithms}, pages 2709--2728. SIAM, 2019.

\bibitem[LN19]{ln19}
Yi~Li and Vasileios Nakos.
\newblock Deterministic sparse {F}ourier transform with an $\ell_{\infty}$
  guarantee.
\newblock {\em arXiv preprint arXiv:1903.00995}, 2019.

\bibitem[LWC13]{lawlor13}
David Lawlor, Yang Wang, and Andrew Christlieb.
\newblock Adaptive sub-linear time {F}ourier algorithms.
\newblock {\em Advances in Adaptive Data Analysis}, 5(01):1350003, 2013.

\bibitem[Man92]{m92}
Yishay Mansour.
\newblock Randomized interpolation and approximation of sparse polynomials.
\newblock In {\em International Colloquium on Automata, Languages, and
  Programming}, pages 261--272. Springer, 1992.

\bibitem[MZIC17]{mzic2017}
Sami Merhi, Ruochuan Zhang, Mark~A Iwen, and Andrew Christlieb.
\newblock A new class of fully discrete sparse {F}ourier transforms: Faster
  stable implementations with guarantees.
\newblock {\em Journal of {F}ourier Analysis and Applications}, pages 1--34,
  2017.

\bibitem[NT09]{nt09}
Deanna Needell and Joel~A Tropp.
\newblock {C}o{S}a{MP}: Iterative signal recovery from incomplete and
  inaccurate samples.
\newblock {\em Applied and computational harmonic analysis}, 26(3):301--321,
  2009.

\bibitem[NV09]{nv09}
Deanna Needell and Roman Vershynin.
\newblock Uniform uncertainty principle and signal recovery via regularized
  orthogonal matching pursuit.
\newblock {\em Foundations of computational mathematics}, 9(3):317--334, 2009.

\bibitem[NV10]{nv10}
D~Needel and R~Vershynin.
\newblock Signal recovery from inaccurate and incomplete measurements via
  regularized orthogonal matching pursuit.
\newblock {\em IEEE Journal of Selected Topics in Signal Processing}, pages
  310--316, 2010.

\bibitem[Phi]{phillips}
Phillips.
\newblock Compressed sense.
\newblock
  \url{https://www.philips.com/healthcare/resources/landing/compressed-sense}.

\bibitem[PR14]{PR14}
Sameer Pawar and Kannan Ramchandran.
\newblock A robust {R}-{FFAST} framework for computing a $k$-sparse $n$-length
  {DFT} in ${O}(k \log n)$ sample complexity using sparse-graph codes.
\newblock In {\em Information Theory (ISIT), 2014 IEEE International Symposium
  on}, pages 1852--1856. IEEE, 2014.

\bibitem[PS15]{ps15}
Eric Price and Zhao Song.
\newblock A robust sparse {F}ourier transform in the continuous setting.
\newblock In {\em Foundations of Computer Science (FOCS), 2015 IEEE 56th Annual
  Symposium on}, pages 583--600. IEEE, 2015.

\bibitem[Rao19]{r19}
Sharavas Rao.
\newblock Improved lower bounds for the restricted isometry property of
  subsampled fourier matrices.
\newblock In {\em arXiv preprint}. \url{https://arxiv.org/pdf/1903.12146.pdf},
  2019.

\bibitem[Rey89]{r89}
George~O Reynolds.
\newblock {\em The New Physical Optics Notebook: Tutorials in {F}ourier
  Optics.}
\newblock ERIC, 1989.

\bibitem[RV08]{rv08}
Mark Rudelson and Roman Vershynin.
\newblock On sparse reconstruction from {F}ourier and {G}aussian measurements.
\newblock {\em Communications on Pure and Applied Mathematics: A Journal Issued
  by the Courant Institute of Mathematical Sciences}, 61(8):1025--1045, 2008.

\bibitem[Sie]{siemens}
Siemens.
\newblock Compressed sensing beyond speed.
\newblock
  \url{https://www.healthcare.siemens.com/magnetic-resonance-imaging/clinical-specialities/compressed-sensing}.

\bibitem[Son17]{s17}
Zhao Song.
\newblock High dimensional {F}ourier transform in the continuous setting.
\newblock In {\em Manuscript}, 2017.

\bibitem[Son19]{song19}
Zhao Song.
\newblock {\em Matrix Theory : Optimization, Concentration and Algorithms}.
\newblock PhD thesis, The University of Texas at Austin, 2019.

\bibitem[TG07]{tg07}
Joel Tropp and Anna~C Gilbert.
\newblock Signal recovery from partial information via orthogonal matching
  pursuit.
\newblock {\em IEEE Trans. Inform. Theory}, 53(12):4655--4666, 2007.

\bibitem[Voe11]{v11}
David~G Voelz.
\newblock {\em Computational {F}ourier Optics: A MATLAB Tutorial (SPIE Tutorial
  Texts Vol. TT89)}.
\newblock SPIE press, 2011.

\end{thebibliography}
\newpage

\newpage
{\hypersetup{linkcolor=black}
\tableofcontents
}
\newpage

\section{Algorithm for $d$-dimensional Sparse Fourier Transform}
In this section, we will give a Fourier sparse recovery algorithm that takes $O(k \log k \log n)$ measurements with ``$\ell_\infty / \ell_2$'' guarantee.
We assume the knowledge of $\mu = \frac{1}{\sqrt{k}} \|\wh{x}_{-k}\|_2$. 
In fact, a constant factor approximation suffices, but we prefer to assume exact knowledge of it in order to simplify exposition. All of the arguments go through in the other case, with minor changes in constants.
We also assume we know $R^*$ so that $R^* \geq \|\wh{x}\|_\infty\ / \mu$. 
We assume that $\log R^* = O(\log n)$.
For larger $\log R^* = O(\poly(n))$, our algorithm will still work, but the decoding time will be worse by a factor of $\frac{\log R^*}{\log n}$.
Note that our assumptions on $\mu$ and $R^*$ are standard.
For example, \cite{ik14} make the same assumption.
We assume that we can measure the signal $x$ in the time domain, and we want to recover the signal $\wh{x}$ in the frequency domain.

In our algorithm, we will use $\mu$ as a threshold for noise, and we will perform $\log R^*$ iterations, where in each iteration the upper bound of $\ell_\infty$ norm of the residual signal (in the frequency domain) shrinks by half.
In Section~\ref{sub:notations}, we give some definitions that will be used in the algorithm.
Then we present our new algorithm for $d$-dimension Fourier sparse recovery in Section~\ref{sub:algorithm}.
In Section~\ref{sub:analysis}, we prove the correctness of the proposed algorithm.

\subsection{Notations} \label{sub:notations}

For a subset of samples (or measurements) $\{x_t\}_{t \in T}$ from the time domain, where $T$ is a list of elements in $[p]^d$, we define $\wh{x}^{[T]}$ in Definition~\ref{def:measurement_notation} as our estimation to $\wh{x}$.

\begin{definition}[Fourier transform of a subset of samples]\label{def:measurement_notation}
Let $x \in \C^{[p]^d}$. For any $T$ which is a list of elements in $[p]^d$, for any $f \in [p]^d$, we define
\begin{align*}
\wh{x}_f^{[T]} = \frac{\sqrt{n}}{|T|} \sum_{t \in T} \omega^{f^\top t} x_t.
\end{align*}
\end{definition}

In order to reuse samples across different iterations where we drive down the upper bound of the residual signal by half, in each iteration after we obtain estimations to heavy hitters (or equivalently large coordinates), instead of subtracting the estimates directly, we need to ``hide'' the randomness leaked by the samples. 
We interpret each estimate (which is a complex number) as a point on a $2$-dimension plane, and hide the randomness by rounding the estimate to the nearest grid point (where the side length of the grid is chosen to be a small constant fraction of the target $\ell_\infty$ norm of the residual signal in the frequency domain), which we call ``projection onto grid''.
In Definition~\ref{def:box_and_grid}, we formally define box and grid, and in Definition~\ref{def:project_to_grid} we define projection to grid.
We illustrate these two definitions in Figure~\ref{fig:box_and_grid}.

\begin{figure}[!t]
  \centering
    \includegraphics[width=0.8\textwidth]{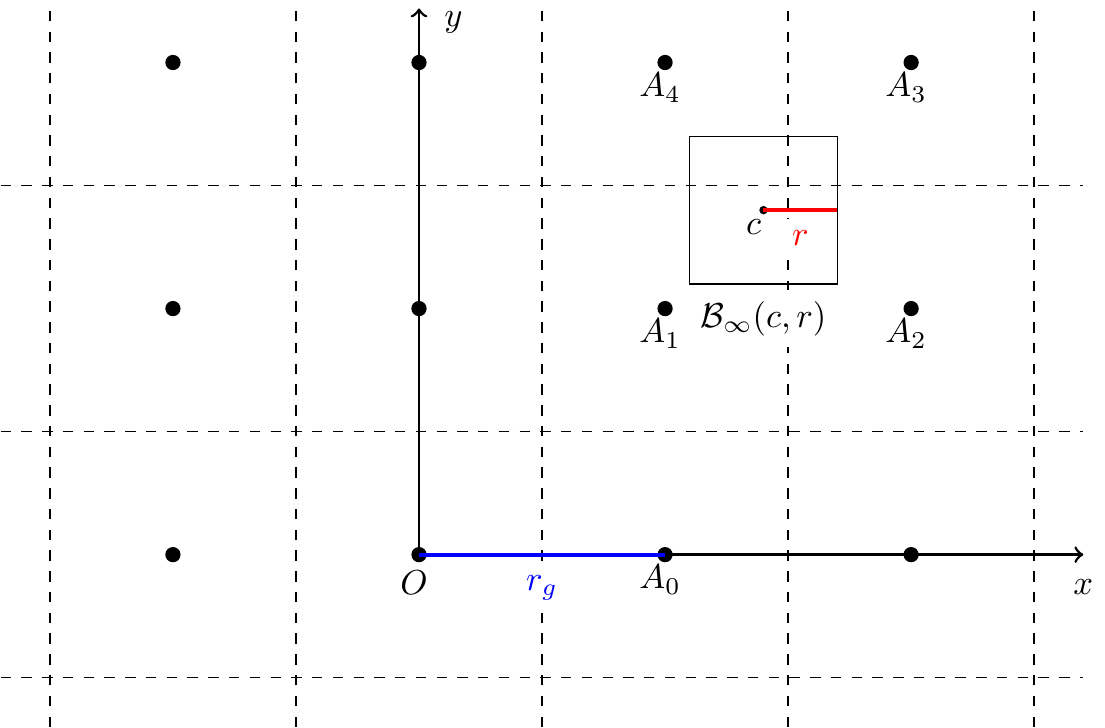}
    \caption{Illustration of box $\mathcal{B}_\infty(c, r)$ and grid $\mathcal{G}_{r_g}$. Box $\mathcal{B}_\infty(c, r)$ refers to all the points in the square centered at $c$ with side length $2r$. Grid $\mathcal{G}_{r_g}$ refers to all the solid round points, and the distance between origin $O$ and $A_0$ is $r_g$. Note that the dashed lines are decision boundaries of the projection $\Pi_{r_g}$, and all the points inside a minimum cell separated by the dashed lines are mapped (by $\Pi_{r_g}$) to the same grid point in $\mathcal{G}_{r_g}$ (which is the center of the cell). In this example we have $\Pi_{r_g}(c) = A_1$ and $\Pi_{r_g}(\mathcal{B}_\infty(c,r)) = \{A_1, A_2, A_3, A_4\}$.}\label{fig:box_and_grid}
\end{figure}

\begin{definition}[box and grid]\label{def:box_and_grid}
For any $c \in \C$ and $r \geq 0$, we define box $\mathcal{B}_\infty(c, r) \subseteq \C$ as
\begin{align*}
\mathcal{B}_\infty(c, r) = \{ c+ x + y \i: x,y \in [-r, r] \}.
\end{align*}

Namely, if we consider complex numbers as points on $2$D plane, box $\mathcal{B}_\infty(c, r)$ refers to $\ell_\infty$ ball with radius $r$ centered at $c$.

For any $r > 0$, we define grid $\mathcal{G}_r \subseteq \C$ as
\begin{align*}
\mathcal{G}_r = \{ x r + y r \i : x, y \in \Z \}.
\end{align*}
\end{definition}

\begin{definition}[projection onto grid]\label{def:project_to_grid}
For any $r > 0$, we define $\Pi_r$ to be a maping from $\C$ to $\mathcal{G}_r$, so that for any $c \in \C$,
\begin{align*}
\Pi_r(c) = \arg\min_{c' \in \mathcal{G}_r} |c - c'|,
\end{align*}
where we break the tie by choosing the one with minimum $|c'|$.
As a natural generalization, For $C \subseteq \C$, we define 
\begin{align*}
\Pi_r(C) = \{\Pi_r(c): c \in C \}.
\end{align*}
\end{definition}

\subsection{Algorithm} \label{sub:algorithm}

We present our new sparse Fourier recovery algorithm in Algorithm~\ref{alg:fourier_sparse_recovery}. 
Its auxiliary function \textsc{LinfinityReduce} is in Algorithm~\ref{alg:linfinity_reduce}.
Important constants are summarized in Table~\ref{tab:parameters}.

In Algorithm~\ref{alg:fourier_sparse_recovery}, we define ``bucket size'' $B = O(k)$ and number of repetitions $R = O(\log n)$.
For each $r \in [R]$, we choose ${\cal T}_r$ to be a list of $B$ independent and uniformly random elements in $[p]^d$.
We will measure $x_t$ for all $t \in \cup_{r \in [R]} {\cal T}_r$, and use \textsc{LinfinityReduce} in Algorithm~\ref{alg:linfinity_reduce} to locate and estimate all the ``heavy hitters'' of the residual signal so that if we substract them then the $\ell_\infty$ norm of the new residual signal shrinks by half.
The input to \textsc{LinfinityReduce} is a signal $x \in \C^{[p]^d}$ in the time domain (but we can only get access to $x_t$ where $t \in \cup_{r \in [R]} {\cal T}_r$), a sparse vector $y \in \C^{[p]^d}$ in the frequency domain that we have recovered so far, and $\nu \geq \mu$ such that $\|\wh{x} - y\|_\infty \leq 2 \nu$ where we will refer $\wh{x} - y$ as the currect residual signal (in the frequency domain).
It is guaranteed that $\textsc{LinfinityReduce}(x, n, y, \{ {\cal T}_r \}_{r=1}^R, \nu)$ returns a $O(k)$-sparse $z$ so that $\|\wh{x} - y - z \| \leq \nu$ with probability $1 - 1 / \poly(n)$.

Algorithm~\ref{alg:fourier_sparse_recovery} will run \textsc{LinfinityReduce} $H = O(\log k)$ times, where in the $h$-th copy it measures $\mathcal{T}^{(h)} = \{{\cal T}_r^{(h)}\}_{r \in [R]}$ for $h \in [H]$.
We denote $\mathcal{T} = \{{\cal T}^{(h)}\}_{r \in [R]}$.
If $\log R^* \leq H$, then we can simply use different $\mathcal{T}^{(h)}$ in different iterations.
In that case $L=1$ and $H = \log R^*$ in Algorithm~\ref{alg:fourier_sparse_recovery}.
We will get $z^{(1)}$ on Line~\ref{lin:z_assign} such that $\|\wh{x} - y^{(0)} - z^{(1)}\|_\infty \leq \mu$ (we will prove in the analysis this holds with probability $1 - 1 / \poly(n)$) where $y^{(0)} = 0$, and return $z^{(1)} + y^{(0)}$ on Line~\ref{lin:return_recovered_signal}.

If $\log R^* > H$, we have to reuse the samples.
We proceed in $L$ iterations (in the loop between Line~\ref{lin:outer_loop_start} and Line~\ref{lin:outer_loop_end} in Algorithm~\ref{alg:fourier_sparse_recovery}), where $L = \log R^* - H + 1$.
For $\ell \in [L]$, as defined in Line~\ref{lin:nu_ell}$, \nu_{\ell} = 2^{-\ell} \mu R^*$ refers to the target $\ell_\infty$ of the residual signal in the $\ell$-th iteration (namely, for $\ell \in [L-1]$ we want to obtain $y^{(\ell)}$ so that $\|\wh{x} - y^{(\ell)}\|_\infty \leq \nu_\ell$).
In the $\ell$-th iteration where $\ell \in [L]$, by using the samples in $\mathcal{T} = \{{\cal T}^{(h)}\}_{h \in {H}}$ (Line~\ref{lin:z_init} to Line~\ref{lin:z_assign}), the algorithm tries to get $z^{(\ell)}$ so that $\|\wh{x} - y^{(\ell-1)} - z^{(\ell)}\|_\infty \leq 2^{1-H} \nu_\ell$.
The intuition on the behavior of Line~\ref{lin:z_init} to Line~\ref{lin:z_assign} is depicted in Figure~\ref{fig:z_ell}.

\begin{figure}[!t]
  \centering
    \includegraphics[width=\textwidth]{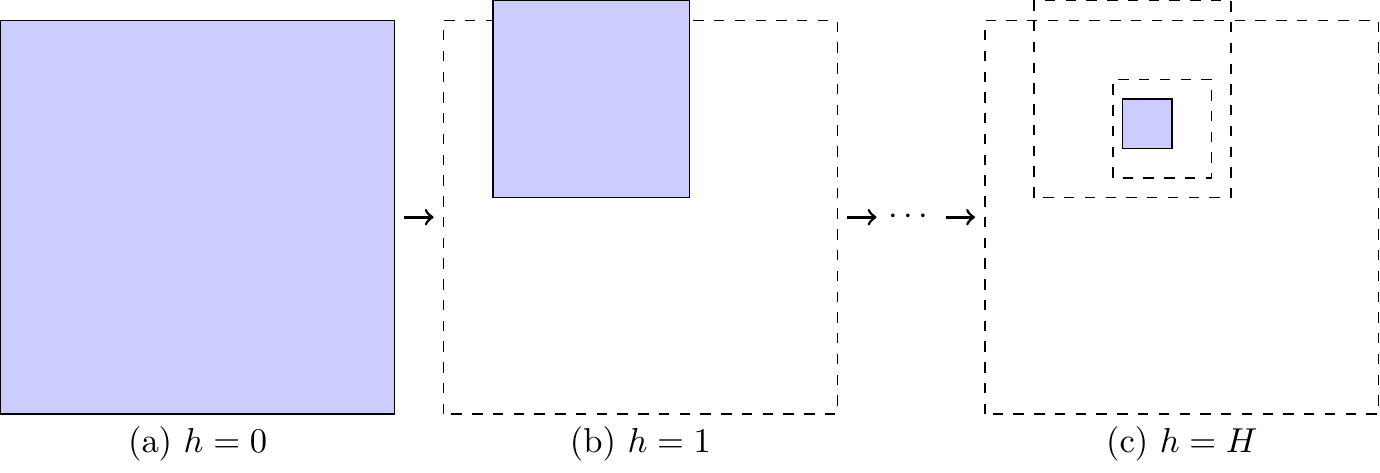}
    \caption{Illustration of the behavior of Line~\ref{lin:z_init} to Line~\ref{lin:z_assign} in Algorithm~\ref{alg:fourier_sparse_recovery}. 
    For any $f \in [p]^d$, we draw box $\mathcal{B}_\infty(y^{(\ell-1)}_f+z_f, 2^{1-h} \nu_\ell)$ after $h$ iterations of the for loop between Line~\ref{lin:inner_loop_start} and Line~\ref{lin:inner_loop_end} in Algorithm~\ref{alg:fourier_sparse_recovery}, where $h \in \{0, 1, \ldots, H\}$.
    Conditioned on \textsc{LinfinityReduce} being correct, for every $h \in \{0, 1, \ldots, H\}$, after $h$-th iteration we have $\wh{x}_f \in \mathcal{B}_\infty(y^{(\ell-1)}_f+z_f, 2^{1-h} \nu_\ell)$.
    When $h = 0$, i.e. before the loop between Line~\ref{lin:inner_loop_start} and Line~\ref{lin:inner_loop_end} starts, we know that $\wh{x}_f \in \mathcal{B}_\infty(y^{(\ell-1)}_f, 2 \nu_\ell)$ as depicted by (a).
    After each iteration in $h$, the radius of the box shrinks by half (and its center might change).
    Finally after $H$ iterations, as depicted by (c), we obtain $z^{(\ell-1)}$ such that $\wh{x}_f \in \mathcal{B}_\infty(y^{(\ell-1)}_f + z^{(\ell)}_f, 2^{1-H} \nu_\ell)$.
    }\label{fig:z_ell}
\end{figure}

If $\ell = L$ the algorithm will return $y^{(L-1)} + z^{(L)}$ as in Line~\ref{lin:return_recovered_signal}; otherwise, the algorithm will try to compute $y^{(\ell)}$ based on $y^{(\ell-1)} + z^{(\ell)}$.
In Line~\ref{lin:repeat_to_find_the_shift} to Line~\ref{lin:until_good}, the algorithm repeatedly samples a uniform random shift $s_\ell \in \mathcal{B}_\infty(0, \alpha \nu_\ell)$ (where $\alpha \in (0,1)$ is a small constant chosen in Table~\ref{tab:parameters}) until the shift is good, where shift $s_\ell$ is good if and only if for each $f \in \supp(y^{(\ell-1)}+z^{(\ell)})$, all the points in $\mathcal{B}_\infty(y^{(\ell-1)}+z^{(\ell)}+s_\ell, 2^{1-H}\nu_\ell)$ (i.e. the box obtained by applying shift $s_\ell$ to the box $\mathcal{B}_\infty(y^{(\ell-1)}+z^{(\ell)}, 2^{1-H}\nu_\ell)$) project to the same grid point in $\mathcal{G}_{\beta \nu_\ell}$.
We depict the process of obtaining the shift $s_\ell$ in Figure~\ref{fig:shift}. 
It is crucial to note that if the shift $s_\ell$ is good and the vector $z^{(\ell)}$ we get is desirable (namely $\|\wh{x} - y^{(\ell-1)} - z^{(\ell)}\|_\infty \leq 2^{1-H} \nu_\ell$), then for each $f \in \supp(y^{(\ell-1)} + z^{(\ell)})$, $\Pi_{\beta \nu_\ell}(y_f^{(\ell-1)} + z_f^{(\ell)} + s_\ell) = \Pi_{\beta \nu_\ell}(\wh{x}_f + s_\ell)$.

\begin{figure}[!t]
  \centering
    \includegraphics[width=\textwidth]{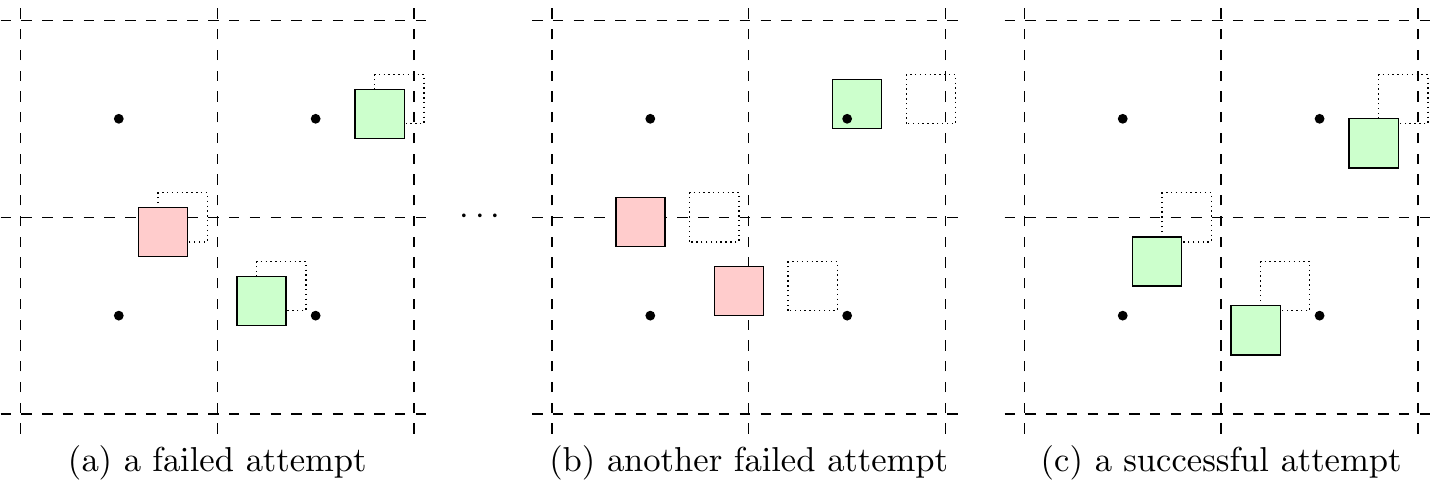}
    \caption{Illustration of the iteration between Line~\ref{lin:repeat_to_find_the_shift} and Line~\ref{lin:until_good} in Algorithm~\ref{alg:fourier_sparse_recovery}.
    The round solid points represent grid points in $\mathcal{G}_{\beta \nu}$, and the dashed lines represent decision boundaries of $\Pi_{\beta \nu_\ell}$.
    In this example we have $|\supp(y^{(\ell-1)} + z^{(\ell)})| = 3$, and the dotted squares represent boxes $\mathcal{B}_\infty(y^{(\ell-1)}_f + z^{(\ell)}_f, 2^{1-H} \nu_\ell)$ for $f \in \supp(y^{(\ell-1)} + z^{(\ell)})$.
    The algorithm repeatedly samples a random shift $s \sim \mathcal{B}_\infty(0, \alpha \nu_\ell)$, until all the shifted boxes $\{\mathcal{B}_\infty(y^{(\ell-1)}_f + z^{(\ell)}_f, 2^{1-H} \nu_\ell) + s\}_{f \in \supp(y^{(\ell-1)} + z^{(\ell)})}$ do not intersect with the dashed lines (i.e. decision boundaries of $\Pi_{\beta \nu_\ell}$).
    In the figure, we color a shifted box in green if it does not intersect with dashed lines, and color in red otherwise.
    After a series of failed attempts from (a) to (b), we finally have a successful attempt in (c).
    }\label{fig:shift}
\end{figure}

On Line~\ref{lin:y_ell_assign}, we assign $\Pi_{\beta \nu_\ell}(y_f^{(\ell-1)} + z_f^{(\ell)} + s_\ell)$ to $y_f^{(\ell)}$.
Because $\beta$ is a small constant, we still have the guarantee that $\|\wh{x} - y^{(\ell)}\|_\infty \leq \nu_\ell$.

\begin{algorithm}[!t]\caption{Fourier sparse recovery by random shift and projection}\label{alg:fourier_sparse_recovery}
\begin{algorithmic}[1]
\Procedure{\textsc{FourierSparseRecovery}}{$x,n,k,\mu,R^*$} \Comment{Theorem~\ref{thm:fourier_sparse_recovery_formal}, $n=p^d$} 
  \State {\bf Require} that $\mu = \frac{1}{\sqrt{k}} \|\wh{x}_{-k}\|_2$ and $R^* \geq \|\wh{x}\|_\infty\ / \mu$ \Comment{$R^*$ is a power of $2$}
	\State $B \leftarrow C_B \cdot k$ \Comment{$C_B$ is a constant defined in Table~\ref{tab:parameters}}
	\State $R \leftarrow C_R \cdot \log n$ \Comment{$C_R$ is a constant defined in Table~\ref{tab:parameters}}
  \State $H \leftarrow \min \{\log k + C_{H}, \log R^* \}$ \Comment{$C_H$ is a constant defined in Table~\ref{tab:parameters}}
  \For{$h = 1 \to H$}
    \For{$r = 1 \to R$}
      \State ${\cal T}_r^{(h)} \leftarrow $a list of $B$ i.i.d elements in $[p]^d$ 
    \EndFor
    \State $\mathcal{T}^{(h)} \leftarrow \{ {\cal T}_r^{(h)} \}_{r = 1} ^ {R}$
  \EndFor \Comment{We will measure $x_t$ for $t \in \cup_{h \in [H], r \in [R]} {\cal T}_r^{(h)}$}
  \State $y^{(0)} \leftarrow \vec 0$ \Comment{$y^{(0)} \in \C^n$}
  \State $L \leftarrow \log R^* - H + 1$
  \For{$\ell=1 \to L$} \label{lin:outer_loop_start}
    \State $\nu_\ell \leftarrow 2^{-\ell}\mu R^*$ \Comment{Target $\ell_\infty$ of the residual signal in iteration $t$} \label{lin:nu_ell}
    \State $z \leftarrow \vec 0$ \Comment{$z$ is a temporary variable used to compute $z^{(\ell)}$} \label{lin:z_init}
    \For{$h=1 \to H$} \label{lin:inner_loop_start}
      \State $ z \leftarrow z + \textsc{LinfinityReduce}(x, n, y^{(\ell-1)}+z, \mathcal{T}^{(h)}, 2^{1-h}\nu_\ell)$
    \EndFor \label{lin:inner_loop_end}
    \State $z^{(\ell)} \leftarrow z$ \Comment{We want $\| \wh{x} - y^{(\ell-1)} - z^{(\ell)}\|_\infty \leq 2^{1-H}\nu_\ell$} \label{lin:z_assign}
    \If {$\ell = L$}
      \State \Return $y^{(L-1)} + z^{(L)}$ \label{lin:return_recovered_signal}
    \EndIf
    \State $a \leftarrow 0$ \Comment{A temporary counter maintained for analysis purpose only}
    \Repeat \label{lin:repeat_to_find_the_shift}
      \State Pick $s_\ell \in \mathcal{B}_\infty(0, \alpha \nu_\ell)$ uniformly at random \Comment{$\alpha \in (0,1)$ is a small constant} \label{lin:alpha} \label{lin:shift}
      \State $a \leftarrow a + 1$ \Comment{$\beta$ in the next line is a small constant where $\alpha < \beta < 0.1$} 
    \Until{$\forall f \in \supp(y^{(\ell-1)}+z^{(\ell)}), |\Pi_{\beta \nu_\ell}(\mathcal{B}_\infty(y^{(\ell-1)}_f + z^{(\ell)}_f + s_\ell, 2^{1-H} \nu_\ell))|=1$} \label{lin:until_good}
    \State $a_\ell \leftarrow a$ \label{lin:a_ell}
    \For{$f \in \supp(y^{(\ell-1)}+z^{(\ell)})$}
      \State $y_f^{(\ell)} \leftarrow \Pi_{\beta \nu_\ell}(y^{(\ell-1)}_f + z^{(\ell)}_f + s_\ell)$ \Comment{We want $\|\wh{x} - y^{(\ell)} \|_\infty \leq \nu_\ell$} \label{lin:y_ell_assign}
    \EndFor
  \EndFor \label{lin:outer_loop_end}
\EndProcedure
\end{algorithmic}
\end{algorithm}

\subsection{Analysis} \label{sub:analysis}

In order to analyze the algorithm, let $S \subseteq [n]$ be top $C_S k$ coordinates of $\wh{x}$ where $C_S = 26$, and let $\ov{S} = [n] \setminus S$.
In order to analyze the performance of \textsc{LinfinityReduce} in Algorithm~\ref{alg:linfinity_reduce}, we need the following definition. 

\begin{definition}[uniform sample]
We say $t$ is sampled from $[p]^d$ uniformly at random if for each $i \in [d]$, we independently sample $t_i$ from $[p]$ uniformly at random. We use $t \sim [p]^d$ to denote it.
\end{definition}

\begin{fact} \label{fac:expectation_0}
Let $\omega = e^{2\pi \i /p}$ where $p$ is any positive integer. For a fixed $f \in [p]^d \setminus \{\vec 0\}$, $\E_{ t \sim [p]^d } [ \omega^{f^\top t} ] = 0$.
\end{fact}
\begin{proof}
Note that $\E_{t \sim [p]^d}[\omega^{f^\top t}] = \prod_{i \in [d]} \E_{t_i \sim [p]}[\omega^{f_i t_i}]$ by the fact that $t_1, \ldots, t_d$ are independent.
Because $f \ne \vec 0$, there exists $i \in [d]$ so that $f_{i} \ne 0$.
We have 
\begin{align*}
\E_{t_{i} \sim [p]} [\omega^{f_{i} t_{i}}] = & ~ \frac{1}{p} \sum_{j=0}^{p-1} (\omega^{f_i})^j \\
= & ~ \frac{1}{p} \cdot \frac{(\omega^{f_i})^0 (1 - (\omega^{f_i})^p)}{1 - \omega^{f_i}} \\
= & ~ 0,
\end{align*}
where the second step follows from the sum of geometry series where $\omega^{f_i} \ne 1$, adn the third step follow from $(\omega^{f_i})^p = e^{2 \pi \i f_i} = 1$.
Therefore, $\E_{ t \sim [p]^d } [ \omega^{f^\top t} ] = 0$.

\end{proof}

We define measurement coefficient as follows:
\begin{definition}[measurement coefficient]\label{def:measurement_coefficient}
For any $f \in [p]^d$ and any $T$ which is a list of elements in $[p]^d$, we define 
\begin{align*}  
\end{align*}
\end{definition}

By definition of $c_f^{[T]}$ and $d$-dimensional Fourier transform, we can decompose $\wh{x}_f^{[T]}$ as follows.

\begin{lemma}[measurement decomposition]\label{lem:measurement_decomposition}
For any $f \in [p]^d$ and any $T$ which is a list of elements in $[p]^d$,
\begin{align*}
\wh{x}_f^{[T]} = \sum_{f' \in [p]^d} c_{f-f'}^{[T]} \wh{x}_{f'}.
\end{align*}
\end{lemma}

\begin{proof}
We have
\allowdisplaybreaks
\begin{align*}
\wh{x}_f^{[T]} = & ~ \frac{\sqrt{n}}{|T|} \sum_{t \in T} \omega^{f^\top t} x_t \\
= & ~ \frac{\sqrt{n}}{|T|} \sum_{t \in T} \omega^{f^\top t} \frac{1}{\sqrt{n}} \sum_{f' \in [p]^d} \omega^{-f'^\top t}  \wh{x}_{f'} \\
= & ~ \frac{\sqrt{n}}{|T|} \sum_{t \in T} \frac{1}{\sqrt{n}} \sum_{f' \in [p]^d} \omega^{(f-f')^\top t} \wh{x}_{f'} \\
= & ~ \sum_{f' \in [p]^d} \left( \frac{1}{|T|} \sum_{t \in T} \omega^{(f-f')^\top t} \right) \wh{x}_{f'}  \\
= & ~ \sum_{f' \in [p]^d} c_{f-f'}^{[T]} \wh{x}_{f'},
\end{align*}
where the first step follow by the definition of $\wh{x}_f^{[T]}$ in Definition~\ref{def:measurement_notation}, second step follows by the definition of inverse $d$-dimensional Fourier transform (see Section~\ref{sec:preliminary}), third and forth step follow by rearranging terms, last step follows by the definition of measurement coefficients $c$ in Definition~\ref{def:measurement_coefficient}.
\end{proof}

\begin{table}[!t]
\centering
\begin{tabular}{ | l | l | l | l | }
    \hline
    {\bf Notation} & {\bf Choice} & {\bf Statement} & {\bf Parameter} \\ \hline
    $C_B$ & $10^6$ & Lemma~\ref{lem:correctness_of_fourier_sparse_recovery} & $B$\\ \hline
    $C_R$ & $10^3$ & Lemma~\ref{lem:E_happen_whp} & $R$ \\ \hline
    $C_H$ & $20$ & Algorithm~\ref{alg:fourier_sparse_recovery} & $H$ \\ \hline
    $\alpha$ & $10^{-3}$ & Algorithm~\ref{alg:fourier_sparse_recovery} Line~\ref{lin:alpha} & shift range \\ \hline
    $\beta$ & $0.04$ & Algorithm~\ref{alg:fourier_sparse_recovery} Line~\ref{lin:until_good} & grid size \\ \hline
    $C_S$ & $26$ & Lemma~\ref{lem:E_happen_whp}, Lemma~\ref{lem:correctness_of_fourier_sparse_recovery} & $|S|$ \\ \hline
  \end{tabular}\caption{Summary of important constants.} \label{tab:parameters}
\end{table}

\begin{table}[!t]
\centering
\begin{tabular}{ | l | l | }
    \hline
    {\bf Lemma} & {\bf Meaning} \\ \hline 
    Lemma~\ref{lem:measurement_decomposition} & measurement decomposition \\ \hline
    Lemma~\ref{lem:coefficient_bound} & properties of coefficient \\ \hline 
    Lemma~\ref{lem:noise_bound} & noise bound \\ \hline
    Lemma~\ref{lem:linfinity_reduce} & guarantee of \textsc{LinfinityReduce} \\ \hline
    Lemma~\ref{lem:random_shift} & property of a randomly shifted box \\ \hline
    Lemma~\ref{lem:E_happen_whp} & event $\mathcal{E}$ happens \\ \hline
    Lemma~\ref{lem:correctness_of_fourier_sparse_recovery} & correctness of our algorithm \\ \hline
 \end{tabular}\caption{Summary of Lemmas.} \label{tab:lemmas}
\end{table}

Let $T$ be a list of i.i.d. samples from $[p]^d$, then the coeffcients $c_f^{[T]}$ defined in Definition~\ref{def:measurement_coefficient} have the following property.

\begin{lemma}[properties of coeffcient $c$]\label{lem:coefficient_bound}
Let $T$ be a list of $B$ independent and uniform random elements in $[p]^d$.
Then we have \\
1. $c_0^{[T]} = 1$. \\
2. For any $f \in [p]^d \setminus \{0\}$, $\E_{T} \left[ |c_f^{[T]}|^2 \right]  = \frac{1}{B}$. \\
3. For any $f,f' \in [p]^d$, $f \ne f'$, $\E_{T} \left[ c_f^{[T]} \cdot \overline{c_{f'}^{[T]}} \right] = 0$.
\end{lemma}
\begin{proof}

{\bf Part 1.} By definition of $c_0^{[T]}$, 
\begin{align*}
c_0^{[T]} = \frac{1}{|T|} \sum_{t \in T} \omega^{0 \cdot t} = 1.
\end{align*}

{\bf Part 2.} Let $T = \{t_1, \ldots, t_B\}$, where $t_i$ is independently and uniformly chosen from $[p]^d$. For any $f \in [p]^d \setminus \{0\}$, 
\begin{align*}
\E_{T} \left[ |c_f^{[T]}|^2 \right] = & ~ \E_{T} \left[ c_f^{[T]} \cdot \overline{c_f^{[T]}} \right] \\
= & ~ \frac{1}{|T|^2} \E_{T} \left[ \sum_{i,j \in [B]} \omega^{f^\top (t_i-t_j)} \right] \\
= & ~ \frac{1}{|T|^2} \left( |T| + \E_{T} \left[ \sum_{i,j\in [B], i \ne j} \omega^{f^\top (t_i-t_j)} \right] \right) \\
= & ~ \frac{1}{|T|} + \frac{1}{|T|^2} \sum_{i,j\in [B], i \ne j} \E_T \left[ \omega^{f^\top (t_i-t_j)} \right] \\
= & ~ \frac{1}{|T|} - \frac{1}{|T|^2} \cdot 0 \\
= & ~ \frac{1}{|T|} = \frac{1}{B},
\end{align*}
where the forth step follows by $\E_T[\omega^{f^\top (t_i - t_j)}] = \E_{t \sim [p]^d} [\omega^{f^\top t}] = 0$, in which $\E_T[\omega^{f^\top (t_i - t_j)}] = \E_{t \sim [p]^d} [\omega^{f^\top t}]$ because $i \ne j$, $t_i, t_j$ are independent and uniformly random distributed in $[p]^d$, $t_i-t_j \sim [p]^d$; $\E_{t \sim [p]^d} [\omega^{f^\top t}] = 0$ follows by by Fact~\ref{fac:expectation_0} and $f$ is not a zero vector.

{\bf Part 3.} For any $f,f' \in [p]^d$, $f \ne f'$,
\begin{align*}
\E_{T} \left[ c_f^{[T]} \cdot \overline{c_{f'}^{[T]}} \right] = & ~ \frac{1}{|T|^2} \E_{T} \left[ \sum_{i,j \in [B]} \omega^{f^\top t_i-f'^\top t_j} \right] \\
= & ~ \frac{1}{|T|^2} \left( \sum_{i,j\in [B], i\ne j} \E_T \left[ \omega^{f^\top t_i - f'^\top t_j} \right] + \sum_{i \in [B]} \E_T \left[ \omega^{(f-f')^\top t_i} \right] \right) \\ 
= & ~ \frac{1}{|T|^2} \left( \sum_{i,j\in [B], i\ne j} \E_{t_i \sim [p]^d} \left[ \omega^{f^\top t_i} \right] \E_{t_j \sim [p]^d} \left[ \omega^{-f'^\top t_j} \right] + \sum_{i \in [B]} \E_{t_i \sim [p]^d} \left[ \omega^{(f-f')^\top t_i} \right] \right) \\ 
= & ~ 0,
\end{align*}
where the second step follows from separating diagonal term and off-diagonal terms, the third step follows from $t_i$ and $t_j$ are independent, the last step follows from Fact~\ref{fac:expectation_0} where $f-f' \ne \vec 0$, and at least one of $f$ and $f'$ is not $\vec 0$.

\end{proof}

Let $T$ be a list of independent and uniformly random elements from $[p]^d$. 
We are going to measure $x_t$ for $t \in T$, and take $\wh{x}^{[T]}_f$ (recall its definition in Definition~\ref{def:measurement_notation}) as estimate to $\wh{x}_f$.
By Lemma~\ref{lem:measurement_decomposition}, $\wh{x}_f^{[T]} = \sum_{f' \in [p]^d} c_{f-f'}^{[T]} \wh{x}_{f'}$.
The following lemma bounds the contribution of coordinates from $V$ where $V \subseteq [p]^d \setminus \{f\}$, namely $| \sum_{f' \in V} c_{f-f'}^{[T]} \wh{x}_{f'}|$.
When analyzing the quality of $\wh{x}^{[T]}_f$ as an approximation to $\wh{x}_f$, we consider coordinates in $V$ as noise, and we usually set $V = [p]^d \setminus \{f\}$.

\begin{lemma}[noise bound]\label{lem:noise_bound}
For any $f \in [p]^d$, $T$ which is a list of $B$ i.i.d. samples from $[p]^d$ and $V \subseteq [n]$ such that $f \not\in V$,
\begin{align*}
\Pr_{T} \left[ \left| \sum_{f' \in V} c_{f-f'}^{[T]} \wh{x}_{f'} \right| \geq \frac{10}{\sqrt{B}} \|\wh{x}_V\|_2 \right] \leq \frac{1}{100}. 
\end{align*}
\end{lemma}

\begin{proof}
First, we can prove that $\E_{T}\left[ \left| \sum_{f' \in V} c_{f-f'}^{[T]} \wh{x}_{f'} \right|^2 \right] = \frac{1}{B} \|\wh{x}_V\|_2^2$, because
\begin{align*}
\E_{T}\left[ \left| \sum_{f' \in V} c_{f-f'}^{[T]} \wh{x}_{f'} \right|^2 \right] = & ~ \E_{T} \left[ \sum_{f_1,f_2 \in V} (c_{f-f_1}^{[T]} \wh{x}_{f_1}) \overline{(c_{f-f_2}^{[T]} \wh{x}_{f_2})} \right] \\
= & ~ \sum_{f_1,f_2 \in V} \E_{T} \left[ c_{f-f_1}^{[T]} \overline{c_{f-f_2}^{[T]}} \right] \wh{x}_{f_1} \overline{\wh{x}_{f_2}} \\
= & ~ \sum_{f' \in V} \E_{T} \left[\left|c_{f-f'}^{[T]}\right|^2\right] |\wh{x}_{f'}|^2 \\
= & ~ \frac{1}{B} \|\wh{x}_V\|_2^2,
\end{align*}
where the third step follows from Lemma~\ref{lem:coefficient_bound} that for $f-f_1 \ne f-f_2$, $\E_{T}\left[ c_{f-f_1}^{[T]} \overline{c_{f-f_2}^{[T]}} \right] = 0$, and the last step follows from $\E_{T} \left[\left|c_{f-f'}^{[T]}\right|^2\right] = 1/B$ in Lemma~\ref{lem:coefficient_bound}.

Then the lemma follows by invoking Chebyshev Inequality and the fact that 
\begin{align*}
\Var_{T} \left[ \left| \sum_{f' \in V} c_{f-f'}^{[T]} \wh{x}_{f'} \right| \right] \leq \E_{T}\left[ \left| \sum_{f' \in V} c_{f-f'}^{[T]} \wh{x}_{f'} \right|^2 \right] = \frac{1}{B} \|\wh{x}_V\|_2^2.
\end{align*}
\end{proof}

\begin{algorithm}[!t]\caption{Procedure for reducing $\ell_\infty$ norm of the residual signal}\label{alg:linfinity_reduce}
\begin{algorithmic}[1]
\Procedure{\textsc{LinfinityReduce}}{$x, n, y, \{ {\cal T}_r \}_{r=1}^R, \nu$} \Comment{Lemma~\ref{lem:linfinity_reduce}}
  \State {\bf Require} that $\| \wh{x} - y \|_\infty \leq 2 \nu$
  \State Let $w$ be inverse Fourier transform of $y$ \Comment{We have $\wh{w} = y$}
  \For{$r=1 \to R$} \label{lin:linfinity_reduce_first_loop_begin}
    \For{$f=1 \to n$} \Comment{Implemented by FFT which takes $O(n \log n)$ time}
      \State $u_{f,r} \leftarrow \frac{\sqrt{n}}{|{\cal T}_r|} \sum_{t \in {\cal T}_r} \omega^{f^\top t} (x_t - w_t)$ \label{lin:u_f_r} \Comment{$\omega = e^{2\pi\i/p}, u_{f,r}=\wh{(x-w)}_f^{[{\cal T}_r]}$}
    \EndFor
  \EndFor \label{lin:linfinity_reduce_first_loop_end}
  \For{$f=1 \to n$}
    \State $\eta = \text{median}_{r\in[R]} \{ u_{f,r} \}$ \Comment{Take the median coordinate-wise} \label{lin:taking_meidan}
    \If{$|\eta| \geq \nu / 2$} \label{lin:if_eta_begin}
      \State $z_f \leftarrow \eta$
    \Else
      \State $z_f \leftarrow 0$
    \EndIf \label{lin:if_eta_end}
  \EndFor
  \State \Return $z$ \Comment{Guarantee $\| \wh{x} - y - z \|_\infty \leq \nu$}
\EndProcedure
\end{algorithmic}
\end{algorithm}

In the next lemma, we show the guarantee of $\textsc{LinfinityReduce}$ in Algorithm~\ref{alg:linfinity_reduce}.

\begin{lemma}[guarantee of $\textsc{LinfinityReduce}$ in Algorithm~\ref{alg:linfinity_reduce}]\label{lem:linfinity_reduce}
Let $x \in \C^{[p]^d}$, and $n = p^d$.
Let $R = C_R \log n$, and $B = C_B k$.
Let $C_B \geq 10^6$ and $C_R \geq 10^3$.
Let $\mu = \frac{1}{\sqrt{k}} \|\wh{x}_{-k}\|_2$, and $\nu \geq \mu$.
For $r \in [R]$, let ${\cal T}_r$ be a list of $B$ i.i.d. elements in $[p]^d$.
Let $z \in \C^n$ denote the output of 
\begin{align*}
\textsc{LinfinityReduce} (x, n, y, \{ {\cal T}_r \}_{r=1}^R, \nu).
\end{align*}
Let $S$ be top $C_S k$ coordinates in $\wh{x}$, where $C_S = 26$.
If $\| \wh{x} - y \|_\infty \leq 2 \nu$, $\supp(y) \subseteq S$ and $y$ is independent from the randomness of $\{{\cal T}_r\}_{r=1}^R$, then with probability $1 - 1 / \poly(n)$ under the randomness of $\{ {\cal T}_r \}_{r=1}^R$, $\| \wh{x} - y - z \|_\infty \leq \nu$ and $\supp(z) \subseteq S$.
Moreover, the running time of $\textsc{LinfinityReduce}$ is $O(n \log^2 n)$.
\end{lemma}

\begin{proof}
  
Note that $\forall f \in \ov{S}$, 
\begin{align*}
|\wh{x}_f| \leq \sqrt{ \frac{\|\wh{x}_{-k}\|_2^2} {C_S k - k} } = \frac{1}{5} \mu ,
\end{align*}
where the last step follows from choice of $C_S$.
  
Let $w$ denote the inverse Fourier transform of $y$.
Note that on Line~\ref{lin:u_f_r} in Algorithm~\ref{alg:linfinity_reduce}, for any $f \in [p]^d$ and $r \in [R]$,
\begin{align*}
u_{f,r} = & ~ \frac{\sqrt{n}}{|{\cal T}_r|} \sum_{t \in {\cal T}_r} \omega^{f^\top t} (x_t - w_t) \\
= & ~ \wh{(x-w)}_f^{[{\cal T}_r]} \\
= & ~ \sum_{f'\in [n]} c_{f-f'}^{[{\cal T}_r]} (\wh{x}_{f'}-y_{f'}),
\end{align*}
where the second step follows by the notation in Definition~\ref{def:measurement_notation}, and the third step follows by Lemma~\ref{lem:measurement_decomposition}.
Therefore,
\begin{align}
\wh{x}_f - y_f = & ~ u_{f,r} - \sum_{f' \in [p]^d \setminus \{f\}} c_{f-f'}^{[{\cal T}_r]}(\wh{x}_{f'} - y_{f'}) , \label{eq:noise_decomposition}
\end{align}
By Lemma~\ref{lem:noise_bound},
\begin{align}\label{eq:noise_bound_application}
\Pr_{{\cal T}_r} \left[ \left| \sum_{f' \in [p]^d \setminus \{f\}} c_{f-f'}^{[{\cal T}_r]} (\wh{x}_{f'} - y_{f'}) \right| \geq \frac{10}{\sqrt{B}} \|(\wh{x} - y)_{[p]^d \setminus \{f\}}\|_2 \right] \leq \frac{1}{100}.
\end{align}

We have
\begin{align}\label{eq:bound_hatx_minus_y_bound_nu}
\frac{10}{\sqrt{B}} \|(\wh{x} - y)_{[p]^d \setminus \{f\}}\|_2 \leq & ~ \frac{10}{\sqrt{B}} \left( \| (\wh{x} - y)_{S \setminus \{f\}} \|_2 + \| (\wh{x} - y)_{\ov{S} \setminus \{f\}} \|_2 \right) \notag \\
\leq & ~ \frac{10}{\sqrt{B}} \left( \|\wh{x} - y\|_\infty \cdot \sqrt{|S|} + \|\wh{x}_{\ov{S} \setminus \{f\}}\|_2 \right) \notag \\ 
\leq & ~ \frac{10}{\sqrt{B}} \left( 2 \nu \cdot \sqrt{26 k} + \sqrt{k} \mu \right) \notag \\
\leq & ~ \frac{1}{100\sqrt{k}} \left( 2 \nu \cdot \sqrt{26 k} + \sqrt{k} \mu \right) \notag \\
< & ~ 0.12 \nu,
\end{align}
where the first step following by triangle inequality, the second step follows by the assumption that $\supp(y) \subseteq S$, the forth step follows by $C_B \geq 10^6$, the last step follows by $\mu \leq \nu$. 

Therefore,
\begin{align*}
\Pr_{{\cal T}_r} [|u_{f,r} - (\wh{x}_f-y_f)| \leq 0.12\nu] = & ~ \Pr_{{\cal T}_r} \left[ \left| \sum_{f' \in [p]^d \setminus \{f\}} c_{f-f'}^{[{\cal T}_r]} (\wh{x}_{f'} - y_{f'}) \right| \leq 0.12 \nu \right] \\
= & ~ 1 - \Pr_{{\cal T}_r} \left[ \left| \sum_{f' \in [p]^d \setminus \{f\}} c_{f-f'}^{[{\cal T}_r]} (\wh{x}_{f'} - y_{f'}) \right| > 0.12 \nu \right] \\
\geq & ~ 1 - \Pr_{{\cal T}_r}\left[ \left| \sum_{f' \in [p]^d \setminus \{f\}} c_{f-f'}^{[{\cal T}_r]} (\wh{x}_{f'} - y_{f'}) \right| \geq \frac{10}{\sqrt{B}} \|(\wh{x} - y)_{[p]^d \setminus \{f\}}\|_2 \right] \\
\geq & ~ 1 - \frac{1}{100},
\end{align*}
where the first step follows by \eqref{eq:noise_decomposition}, the third step follows by \eqref{eq:bound_hatx_minus_y_bound_nu}, and the last step follows by \eqref{eq:noise_bound_application}.
 
Thus we have 
\begin{align*}
\Pr_{{\cal T}_r}[ u_{f,r} \in \mathcal{B}_\infty(\wh{x}_f - y_f, 0.12 \nu)] \geq \Pr_{{\cal T}_r} [|u_{f,r} - (\wh{x}_f-y_f)| \leq 0.12\nu] \geq 1 - \frac{1}{100}.
\end{align*}
Let $\eta_f = \text{median}_{r \in [R]} u_{f,r}$ as on Line~\ref{lin:taking_meidan} in Algorithm~\ref{alg:linfinity_reduce}.
By Chernoff bound, with probability $1 - 1/\poly(n)$, more than $\frac{1}{2} R$ elements in $\{u_{f,r}\}_{r=1}^{R}$ are contained in box $\mathcal{B}_\infty(\wh{x}_f - y_f, 0.12 \nu)$, so that $\eta_f \in \mathcal{B}_\infty(\wh{x}_f - y_f, 0.12 \nu)$.

Therefore, we have 
\begin{align*}
\Pr[|\eta_f - (\wh{x}_f - y_f)|  \leq 0.17 \nu] \geq 
\Pr[|\eta_f - (\wh{x}_f - y_f)| \leq \sqrt{2} \cdot 0.12 \nu ] \geq 1 - 1/\poly(n).
\end{align*}
Let $E$ be the event that for all $f \in [p]^d$, $|\eta_f - (\wh{x}_f - y_f)| \leq 0.17 \nu$.
By a union bound over $f \in [p]^d$, event $E$ happens with probability $1 - 1/\poly(n)$.
In the rest of the proof, we condition on event $E$.

({\bf Case 1}) For $f \in \ov{S}$, note that 
\begin{align*}
|\eta_f| \leq 0.17 \nu + |\wh{x}_f - y_f| = 0.17 \nu + |\wh{x}_f| \leq 0.17 \nu + 0.2\nu = 0.37 \nu.
\end{align*}

According to the if statement between Line~\ref{lin:if_eta_begin} and Line~\ref{lin:if_eta_end} in Algorithm~\ref{alg:linfinity_reduce}, $z_f$ will be assigned $0$.
Thus $\supp(z) \subseteq S$.
In addition, $|\wh{x}_f - y_f - z_f| = |\wh{x}_f| \leq \mu \leq \nu$.

({\bf Case 2}) For $f \in S$, we have two cases.
We prove that $|(\wh{x}_f - y_f) - z_f| \leq \nu$ for both cases.

({\bf Case 2.1}) $|\eta_f| < 0.5 \nu$.
$z_f$ is assigned $0$. 
Because 
\begin{align*}
|\eta_f - (\wh{x}_f - y_f)| \leq 0.17 \nu, ~~~ |\wh{x}_f - y_f| \leq |\eta_f| + 0.17 \nu \leq 0.67 \nu.
\end{align*} 
Therefore, 
\begin{align*}
|(\wh{x}_f - y_f) - z_f| \leq 0.67 \nu \leq \nu.
\end{align*}

({\bf Case 2.2}) $|\eta_f| \geq 0.5 \nu$.
$z_f$ is assigned $\eta_f$. We have
\begin{align*}
|(\wh{x}_f - y_f) - z_f| = & ~ |(\wh{x}_f - y_f) - \eta_f| \leq 0.17 \nu \leq \nu.
\end{align*} 

We thus have obtained that with probability $1 - 1/\poly(n)$, $\| (\wh{x} - y) - z \|_\infty \leq \nu$ and $\supp(z) \subseteq S$.

The running time of \textsc{LinfinityReduce} is dominated by the loop between Line~\ref{lin:linfinity_reduce_first_loop_begin} and Line~\ref{lin:linfinity_reduce_first_loop_end}, which takes $O(R \cdot n \log n) = O(n \log^2 n)$ by FFT.
\end{proof}

For a given box $\mathcal{B}_\infty(c, r)$ and grid $\mathcal{G}_{r_g}$, we say a shift $s \in \C$ is good if after applying the shift, all the points in the shifted box $\mathcal{B}_\infty(c, r)+s$ are mapped to the same point by $\Pi_{r_g}$ (recall that $\Pi_{r_g}$ projects any point to the nearst grid point in $\mathcal{G}_{r_g}$). 
We formulate the notion of a good shift in the following definition, and illustrate in Figure~\ref{fig:good_shift}.

\begin{definition}[good shift] \label{def:good_shift}
For any $r_g$, $r_b$, and any $c \in \C$, we say shift $s \in \C$ is a good shift if
\begin{align*}
\left| \Pi_{r_g} ( \mathcal{B}_{\infty} ( c , r_b ) + s ) \right| = 1 .
\end{align*}
\end{definition}

\begin{figure}[!t]
  \centering
    \includegraphics[width=\textwidth]{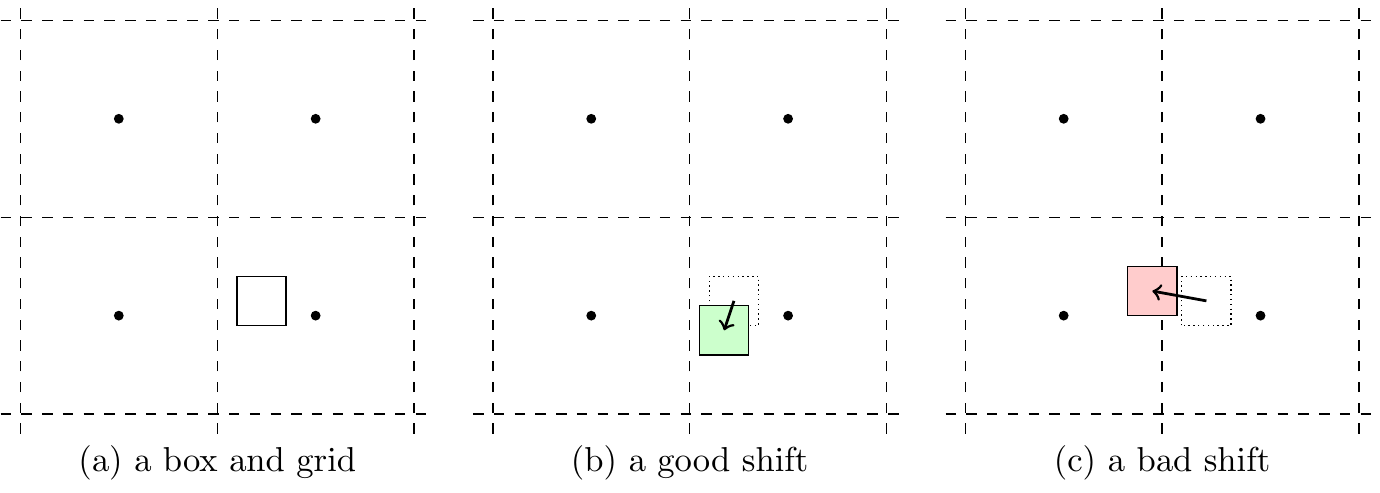}
    \caption{Illustration of good and bad shifts in Definition~\ref{def:good_shift}. 
    In (a), the small square represents box $\mathcal{B}_\infty(c, r_b)$, and the dashed lines represent the decision boundary of $\Pi_{r_g}$.
    The arrows in (b) and (c) represent two different shifts, where the shift in (b) is an example of good shift, since the shifted box does not intersect with the decision boundaries of $\Pi_{r_g}$, while the shift in (c) is an example of bad shift, since the shifted box intersects with the decision boundaries of $\Pi_{r_g}$.
    }\label{fig:good_shift}
\end{figure}

The following lemma intuitively states that if we take a box of radius $r_b$ (or equivalently, side length $2r_b$) and shift it randomly by an offset in $\mathcal{B}_\infty(0, r_s)$ (or equivalently, $[-r_s, r_s] \times [-r_s, r_s]$) where $r_s \geq r_b$, and next we round everyone inside that shifted box to the closest point in $G_{r_g}$ where $r_g \geq 2 r_s$, then with probability at least $(1 - r_b/r_s)^2$ everyone will be rounded to the same point.
In other words, let $s \sim \mathcal{B}_\infty(0, r_s)$, for box $\mathcal{B}_\infty(c, r_b)$ and grid $\mathcal{G}_{r_g}$, $s$ is a good shift with probability at least $(1 - r_b/r_s)^2$.
We illustrate the lemma in Figure~\ref{fig:random_shift_lemma}.

\begin{figure}[!t]
  \centering
    \includegraphics[width=0.8\textwidth]{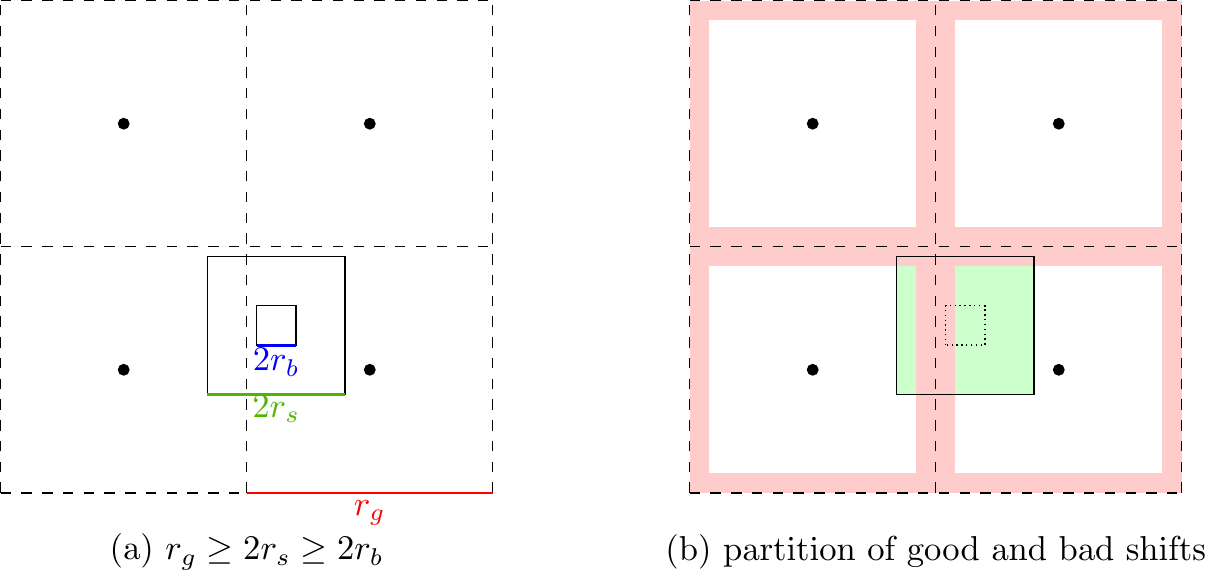}
    \caption{Illustration of Lemma~\ref{lem:random_shift}. In (a) the smallest square represents box $\mathcal{B}_\infty(c, r_b)$, the medium-sized square represents $\mathcal{B}_\infty(c, r_s)$, and the dashed lines represent decision boundaries of $\Pi_{r_g}$. 
    Note that for $s \sim \mathcal{B}_\infty(0, r_s)$, the center of the shifted box $s + \mathcal{B}_\infty(c, r_b)$ is $s+c \sim \mathcal{B}_\infty(c, r_s)$.
    Shift $s$ is good (recall in Definition~\ref{def:good_shift}) for box $\mathcal{B}_\infty(c, r_b)$ and grid $\mathcal{G}_{r_g}$ if and only if the distance between $s+c$ and decision boundaries of $\Pi_{r_g}$ is greater than $r_b$.
    In (b), we draw in red the set of points which are within distance at most $r_b$ to the decision boundaries of $\Pi_{r_g}$. Then in (b) the red part inside $\mathcal{B}_\infty(c, r_s)$ corresponds to bad shifts (plus $c$), and the green part corresponds to good shifts (plus $c$). 
    Intuitively, the fraction of the green part is at least $(1 - r_b / r_s)^2$ because the vertical red strips can cover a width of at most $2r_b$ on the $x$-axis of $\mathcal{B}_\infty(c, r_s)$ (whose side length is $2r_s$), and the horizontal red strips can cover a width of at most $2r_b$ on the $y$-axis.}\label{fig:random_shift_lemma}
\end{figure}

\begin{lemma}[property of a randomly shifted box]\label{lem:random_shift} 
For any $r_g,r_s,r_b$ so that $r_g/2 \geq r_s \geq r_b > 0$ and any $c \in \C$, let $s \in \C$ be uniform randomly chosen in $\mathcal{B}_\infty(0, r_s)$, then 
\begin{align*}
\Pr_{s \sim \mathcal{B}_\infty(0, r_s)} \Bigg[ \bigg| \Pi_{r_g}( \mathcal{B}_\infty(c, r_b) +s) \bigg| = 1 \Bigg] \geq \left( 1 - \frac{r_b}{r_s} \right)^2,
\end{align*}
where we refer $r_g,r_s,r_b$ as the radius of grid, shift and box respectively, and we use notation $C + s$ to refer to $\{c+s: c \in C\}$.
\end{lemma}

\begin{proof}
We consider complex numbers as points in $2$D plane, where the real part is the coordinate on $x$-axis, and the imaginary part is the coordinate on $y$-axis.
Note that the ``decision boundary'' of projection $\Pi_{r_g}$ from $\C$ onto grid $\mathcal{G}_{r_g}$ consists of vertical lines of form $x = (m + \frac{1}{2}) r_g$ and horizontal lines of form $y = (m + \frac{1}{2}) r_g$, where $m \in \Z$.
$\left| \Pi_{r_g}( \mathcal{B}_\infty(c, r_b) +s) \right| = 1$ if and only if the shifted box $\mathcal{B}_\infty(c, r_b) +s$ does not intersect with the ``decision boundary''.

Let $s = s_x + s_y \i$ and $c = c_x + c_y \i$. 
Then the shifted box does not intersect with the ``decision boundary'' if and only if both intervals 
\begin{align*}
[c_x - r_b + s_x, c_x + r_b + s_x] \text{~and~} [c_y - r_b + s_y, c_y + r_b + s_y]
\end{align*}
do not intersect with $\{(m+\frac{1}{2})r_g: m \in \Z \}$.
The probability of each one is at least $1 - \frac{r_b}{r_s}$, and two events are independent.
Therefore, we get the claimed result.
\end{proof}

In the following, we define event $\mathcal{E}$, which is a sufficient condition for the correctness of Algorithm~\ref{alg:fourier_sparse_recovery}.
Event $\mathcal{E}$ consists of three parts. 
Part~1 of $\mathcal{E}$ is used to prove that $a_\ell \leq 10 
\log n$ for $\ell \in [L-1]$ on Line~\ref{lin:a_ell} in Algorithm~\ref{alg:fourier_sparse_recovery}.
Part~2 and Part~3 of $\mathcal{E}$ are used to prove that Line~\ref{lin:nu_ell} to Line~\ref{lin:z_assign} in Algorithm~\ref{alg:fourier_sparse_recovery} give a desirable $z^{(\ell)}$ for $\ell \in [L]$.

\begin{definition}[sufficient condition for the correctness of Algorithm~\ref{alg:fourier_sparse_recovery}] \label{def:sufficient_event}
For input signal $x \in \C^n$, let $\mu = \frac{1}{\sqrt{k}} \|\wh{x}_{-k}\|_2$ and $R^*$ is an upper bound of $\|\wh{x}\|_\infty / \mu$.
Let $S$ be top $C_S k$ coordinates in $\wh{x}$.
Let $H = \min \{ \log k + C_H, \log R^* \}$, and $L = \log R^* - H + 1$.
For $\ell \in [L]$, let $\nu_\ell = 2^{-\ell} \mu R^*$.
For $\ell \in [L-1]$, let $s_\ell^{(a)}$ be the $a$-th uniform randomly sampled from $\mathcal{B}_\infty(0, \alpha \nu_\ell)$ as appeared on Line~\ref{lin:shift} in Algorithm~\ref{alg:fourier_sparse_recovery} (i.e. $s_\ell^{(1)}, \ldots, s_\ell^{(a_\ell)}$ are sampled, and $s_\ell^{(a_\ell)}$ is the first that satisfies the condition on Line~\ref{lin:until_good}).
For the sake of analysis, we assume that Algorithm~\ref{alg:fourier_sparse_recovery} actually produces an infinite sequence of shifts $s_\ell^{(1)}, s_\ell^{(2)}, \ldots$, and chooses the smallest $a_\ell$ so that $s_\ell^{(a_\ell)}$ satisfies $\forall f \in \supp(y^{(\ell-1)}+z^{(\ell)}), |\Pi_{\beta \nu_\ell}(\mathcal{B}_\infty(y^{(\ell-1)}_f + z^{(\ell)}_f + s_\ell^{(a_\ell)}, 2^{1-H} \nu_\ell))|=1$ on Line~\ref{lin:until_good}.

For $\ell \in [L-1]$, we define random variable $a'_\ell$ to be the smallest $a'$ such that for all $f \in S$,
 \begin{align*}
 \left| \Pi_{\beta \nu_\ell} \left(\mathcal{B}_\infty(\wh{x}_f + s_\ell^{(a')}, 2^{3-H} \nu_\ell) \right) \right| = 1.
 \end{align*}

We define event $\mathcal{E}$ to be all of the following events hold. \\
1. For all $\ell \in [L-1]$, $a'_\ell \leq 10 \log n$. \\
2. For $\ell = 1$, if we execute Line~\ref{lin:nu_ell} to Line~\ref{lin:z_assign} in Algorithm~\ref{alg:fourier_sparse_recovery} with $y^{(0)} = 0$, we get $z^{(1)}$ such that $\|\wh{x} - z^{(1)}\|_\infty \leq 2^{1-H} \nu_1$ and $\supp(z^{(1)}) \subseteq S$. \\
3. For all $\ell \in \{2, \ldots, L\}$, for all $a \in [10 \log n]$, if we execute Line~\ref{lin:nu_ell} to Line~\ref{lin:z_assign} in Algorithm~\ref{alg:fourier_sparse_recovery} with $y^{(\ell-1)} = \xi$ where
\begin{align*}
\xi_f = 
\begin{cases} 
  \Pi_{\beta \nu_\ell}(\wh{x}_f + s_{\ell-1}^{(a)}), & \text{~if~} f \in S; \\ 
  0, & \text{~if~} f \in \ov{S}.
\end{cases} 
\end{align*}
then we get $z^{(\ell)}$ such that $\|\wh{x} - y^{(\ell-1)} - z^{(\ell)}\|_\infty \leq 2^{1-H} \nu_\ell$  and $\supp(y^{(\ell-1)}+z^{(\ell)}) \subseteq S$.

\end{definition}

In the following, we will prove that for fixed $x$, under the randomness of $\{s_\ell^{(a)}\}_{\ell \in[L-1], a\in\{1,\ldots\}}$ and $\mathcal{T} = \{\mathcal{T}^{(h)}\}_{h\in[H]}$, event $\mathcal{E}$ (defined in Definition~\ref{def:sufficient_event}) happens with probability at least $1 - 1 / \poly(n)$.
Moreover, we will prove that event $\mathcal{E}$ is a sufficient condition for the correctness of Algorithm~\ref{alg:fourier_sparse_recovery}.
Namely, conditioned on event $\mathcal{E}$, Algorithm~\ref{alg:fourier_sparse_recovery} gives a desirable output. 

\begin{lemma}[event $\mathcal{E}$ happens with high probability]\label{lem:E_happen_whp}
Let $\mathcal{E}$ in Definition~\ref{def:sufficient_event}. For any fixed $x \in \C^n$, under the randomness of shifts $\{s_\ell^{(a)}\}_{\ell \in [L-1], a \in \{1,\ldots\}}$ and $\mathcal{T} = \{ \mathcal{T}^{(h)}\}_{h \in [H]}$, 
\begin{align*}
\Pr[\mathcal{E}] \geq 1 - 1/\poly(n).
\end{align*}
\end{lemma}

\begin{proof}
We bound the failure probability of each parts in event $\mathcal{E}$ respectively as follows, and $\Pr[\mathcal{E}] \geq 1 - 1/\poly(n)$ follows by a union bound.

{\bf Part 1.} If $H = \log R^*$, then $L = 1$ and it is trivially true that ``for all $\ell \in [L-1], a'_\ell \leq 10 \log n$''. Otherwise, we have $H = \log k + C_H$. By Lemma~\ref{lem:random_shift}, for any $\ell \in [L-1]$, for any $f \in S$,
\begin{align*}
\Pr_{s \sim \mathcal{B}_\infty(0, \alpha \nu_\ell)} \Bigg[ \bigg| \Pi_{\beta \nu_\ell} \left( \mathcal{B}_\infty \left( \wh{x}_f, 2^{3-H}\nu_\ell \right) +s \right) \bigg| = 1 \Bigg] \geq \left( 1 - \frac{2^{3-H}\nu_\ell}{\alpha \nu_\ell} \right)^2 = \left( 1 - \frac{2^{3-H}}{\alpha} \right)^2 ,
\end{align*}
where $(1 - 2^{3-H} / \alpha)^2 \geq 1 - 2^{4- C_H - \log k} / \alpha \geq 1 - \frac{1}{100 k}$ by our choice of $\alpha$ and $C_H$ in Table~\ref{tab:parameters}.

For each $\ell \in [L-1]$, by a union bound over all $f$ in $S$, the probability is at least $ 1 - \frac{C_S k}{100k} = 1 - \frac{26k}{100k} \geq \frac{1}{2}$ that for all $f \in S$, $|\Pi_{\beta \nu_\ell}(\mathcal{B}_\infty(\wh{x}_f + s, 2^{3-H} \nu_\ell))| = 1$ where $s \sim \mathcal{B}_\infty(0, \alpha \nu_\ell)$. Formally, we get
\begin{align*}
\Pr_{s \sim \mathcal{B}_\infty(0, \alpha \nu_\ell)} \Bigg[  \bigg| \Pi_{\beta \nu_\ell} \left( \mathcal{B}_\infty \left( \wh{x}_f, 2^{3-H}\nu_\ell \right) +s \right) \bigg| = 1, \forall f \in S \Bigg] \geq 1/2.
\end{align*}

Therefore, by definition of $a'_\ell$ in Definition~\ref{def:sufficient_event},
\begin{align*}
\Pr[a'_\ell \leq 10\log n] \geq 1 - (1 / 2)^{10 \log n} = 1 - 1 / n^{10}.
\end{align*}
By a union bound over all $\ell \in [L-1]$, the probability is at least $1 - L / n^{10} = 1 - 1 / \poly(n)$ that for all $\ell \in [L-1]$, $a'_\ell \leq 10 \log n$.

{\bf Part 2.} By Lemma~\ref{lem:linfinity_reduce} and a union bound over all $h \in [H]$, the failure probability is at most $H / \poly(n) = 1 / \poly(n)$, where $H = O(\log k)$ and so $H / \poly(n)$ is still $1 / \poly(n)$.

{\bf Part 3.} For each $\ell \in \{2, \ldots, L\}$ and $a \in [10 \log n]$, similar to the above argument, each has failure probability at most $1 / \poly(n)$.
By a union bound, the failure probability is at most 
\begin{align*}
(L - 1) \cdot (10 \log n) / \poly(n) = 1 / \poly(n).
\end{align*}

\end{proof}

In the following lemma, we show that if event $\mathcal{E}$ (defined in Definition~\ref{def:sufficient_event}) happens, then Algorithm~\ref{alg:fourier_sparse_recovery} gives a desirable output.

\begin{lemma}[correctness of Algorithm~\ref{alg:fourier_sparse_recovery} conditioned on $\mathcal{E}$]\label{lem:correctness_of_fourier_sparse_recovery}
Let $n = p^d$, and let $k \in [n]$.
Let $x \in \C^{n}$ be input signal.
Let $\mu = \frac{1}{\sqrt{k}} \|\wh{x}_{-k}\|_2$. 
Let $R^* \geq \|\wh{x}\|_\infty\ / \mu$ and $R^*$ is a power of $2$. 
Let $H = \min \{\log k + C_H, \log R^*\}$.
Let $L = \log R^* - H + 1$.
For $\ell \in [L-1]$, let $y^{(\ell)}$ be the vector obtained on Line~\ref{lin:y_ell_assign} of Algorithm~\ref{alg:fourier_sparse_recovery}.
For $\ell \in [L]$, let $z^{(\ell)}$ be the vector obtained on Line~\ref{lin:z_assign}.
Note that $y^{(0)} = 0$, and $y^{(L-1)}+z^{(L)}$ is the output of $\textsc{FourierSparseRecovery}(x,n,k,R^*,\mu)$ in Algorithm~\ref{alg:fourier_sparse_recovery}.
Conditioned on the event $\mathcal{E}$ (defined in Definition~\ref{def:sufficient_event}) happens, we have 
\begin{align*}
\| \wh{x} - y^{(L-1)} - z^{(L)} \|_\infty \leq \frac{1}{\sqrt{k}} \|\wh{x}_{-k}\|_2.
\end{align*}
\end{lemma}

\begin{proof}
We first discuss the case that $H = \log R^*$.
In that case, $L = 1$. 
Conditioned on the event $\mathcal{E}$ (Part~2 of $\mathcal{E}$), $z^{(1)}$ obtained through Line~\ref{lin:nu_ell} to Line~\ref{lin:z_assign} in Algorithm~\ref{alg:fourier_sparse_recovery} satisfies $\|\wh{x} - z^{(1)}\|_\infty \leq 2^{1-H} \nu_1 = 2^{1-H} (2^{-1}\mu R^*) = \mu$.

In the rest of the proof, we discuss the case that $H > \log R^*$.
For $\ell \in [L]$, let $\nu_\ell = 2^{-\ell} \mu R^*$.
For $\ell \in [L-1]$, let $s_\ell^{(a_\ell)} \in \mathcal{B}_\infty(0, \alpha \nu_\ell)$ denote the first $s_\ell^{(a)}$ on Line~\ref{lin:shift} in Algorithm~\ref{alg:fourier_sparse_recovery} such that for all $f \in \supp(y^{(\ell-1)}+z^{(\ell)})$,
\begin{align*}
\left| \Pi_{\beta \nu_\ell} \left( \mathcal{B}_\infty \left( y_f^{(\ell-1)} + z^{(\ell)}_f + s_\ell^{(a)}, 2^{1-H} \nu_\ell \right) \right) \right|=1.
\end{align*}
For $\ell \in [L-1]$, we define $\xi^{(\ell)} \in \C^{[p]^d}$ as follows
\begin{align*}
\xi_f^{(\ell)} = 
\begin{cases} 
  \Pi_{\beta \nu_\ell}(\wh{x}_f + s_\ell^{(a_\ell)}), & \text{~if~} f \in S; \\ 
  0, & \text{~if~} f \in \ov{S}.
\end{cases}
\end{align*}
We also define $\xi^{(0)} = 0$, $s_0^{(a)} = 0$ for $a \in \{1,\ldots\}$ and $a_0 = 1$.

({\bf Goal: Inductive Hypothesis}) We are going to prove that conditioned on event $\mathcal{E}$ (defined in Definition~\ref{def:sufficient_event}), for all $\ell \in \{0,\ldots, L-1\}$, 
\begin{align*}
y^{(\ell)} = \xi^{(\ell)} \text{~~~and~~~} a_\ell \leq 10 \log n.
\end{align*}

({\bf Base case}) Note that $y^{(0)} = \xi^{(0)} = 0$ and $a_0 = 1 \leq 10 \log n$. 

({\bf Inductive step}) We will prove that conditioned on event $\mathcal{E}$, if $y^{(\ell-1)} = \xi^{(\ell-1)}$ and $a_{\ell-1} \leq 10 \log n$ for $\ell \in [L-1]$, then $y^{(\ell)} = \xi^{(\ell)}$ and $a_\ell \leq 10 \log n$.

({\bf Proving $a_\ell \leq 10 \log n$}) Conditioned on event $\mathcal{E}$ (if $L=1$ then from Part~2 of $\mathcal{E}$, otherwise from Part~3 of $\mathcal{E}$ and by the fact that $a_{\ell-1} \leq 10 \log n$), $z^{(\ell)}$ obtained through Line~\ref{lin:nu_ell} to Line~\ref{lin:z_assign} in Algorithm~\ref{alg:fourier_sparse_recovery} satisfies $\|\wh{x} - \xi^{(\ell-1)} - z^{(\ell)}\|_\infty \leq 2^{1-H} \nu_\ell$ and $\supp(z^{(\ell)}) \subseteq S$.
Namely, for all $f \in [p]^d$, $\xi_f^{(\ell-1)} + z^{(\ell)}_f \in \mathcal{B}_\infty(\wh{x}_f, 2^{1-H} \nu_\ell)$.
Recall the definition of $a_\ell'$ in Definition~\ref{def:sufficient_event}.
We can prove that $a_\ell \leq a_\ell'$ because if for all $f \in S$, 
\begin{align*}
\left| \Pi_{\beta \nu_\ell} \left(\mathcal{B}_\infty \left( \wh{x}_f + s_\ell^{(a'_\ell)}, 2^{3-H} \nu_\ell \right) \right) \right| = 1,
\end{align*}
then for all $f \in \supp(y^{(\ell-1)}+z^{(\ell)})$, 
\begin{align*}
\left| \Pi_{\beta \nu_\ell} \left( \mathcal{B}_\infty \left( y^{(\ell-1)}_f + z^{(\ell)}_f + s_\ell^{(a'_\ell)}, 2^{1-H} \nu_\ell \right) \right) \right|=1
\end{align*}
where 
\begin{align*}
\mathcal{B}_\infty \left( y^{(\ell-1)}_f + z^{(\ell)}_f + s_\ell^{(a'_\ell)}, 2^{1-H} \nu_\ell \right) \subseteq \mathcal{B}_\infty \left( \wh{x}_f + s_\ell^{(a'_\ell)}, 2^{3-H} \nu_\ell \right)
\end{align*}
which follows by $\xi_f^{(\ell-1)} + z^{(\ell)}_f \in \mathcal{B}_\infty(\wh{x}_f, 2^{1-H} \nu_\ell)$.
Therefore, conditioned on $\mathcal{E}$ (Part~1 of $\mathcal{E}$), $a_\ell \leq a_\ell' \leq 10 \log n$.

({\bf Proving $y^{(\ell)}_f = \xi^{(\ell)}_f$}) For $f \in [p]^d$, we will prove that $y^{(\ell)}_f = \xi^{(\ell)}_f$ in two cases.

({\bf Case 1}) If $f \in \supp(y^{(\ell-1)} + z^{(\ell)}) \subseteq S$. We have 
\begin{align*}
y_f^{(\ell)} 
= & ~ \Pi_{\beta \nu_\ell} \left( y^{(\ell-1)}_f + z^{(\ell)}_f + s_\ell^{(a_\ell)} \right) \\
= & ~ \Pi_{\beta \nu_\ell} \left( \xi^{(\ell-1)}_f + z^{(\ell)}_f + s_\ell^{(a_\ell)} \right).
\end{align*} 
Because $\xi_f^{(\ell-1)} + z^{(\ell)}_f \in \mathcal{B}_\infty(\wh{x}_f, 2^{1-H} \nu_\ell)$, we have $\xi^{(\ell-1)}_f + z^{(\ell)}_f + s_\ell^{(a_\ell)} \in \mathcal{B}_\infty(\wh{x}_f + s_\ell^{(a_\ell)}, 2^{1-H} \nu_\ell)$.
By the choice of $s_\ell^{(a_\ell)}$, $\Pi_{\beta \nu_\ell}(\xi^{(\ell-1)}_f + z^{(\ell)}_f + s_\ell^{(a_\ell)}) = \Pi_{\beta\nu_\ell}(\wh{x}_f + s_\ell^{(a_\ell)})$.
Thus $y_f^{(\ell)} = \xi_f^{(\ell)}$.

({\bf Case 2}) If $f \not\in \supp(y^{(\ell-1)} + z^{(\ell)})$. We have $y_f^{(\ell)} = 0$.
Because $\xi_f^{(\ell-1)} + z^{(\ell)}_f \in \mathcal{B}_\infty(\wh{x}_f, 2^{1-H} \nu_\ell)$, we have $|\wh{x}_f| < 2^{2 - H} \nu_\ell < 0.1 \beta \nu_\ell$ by our choice of $H$.
We can easily prove that $\xi_f^{(\ell)} = 0 = y_f^{(\ell)}$ in the following two cases: 

({\bf Case 2.1}) If $f \in S$, we have $\xi_f^{(\ell)} = \Pi_{\beta \nu_\ell}(\wh{x}_f + s_\ell^{(a_\ell)}) = 0$ because
\begin{align*}
|\wh{x}_f| + |s_\ell^{(a_\ell)}| < 0.1 \beta \nu_\ell + 2 \alpha \nu_\ell < 0.5 \beta \nu_\ell.
\end{align*}

({\bf Case 2.2}) If $f \in \ov{S}$,  $\xi_f^{(\ell)} = 0$ by definition of $\xi^{(\ell)}$.

Therefore, for all $\ell \in [L-1]$, $y^{(\ell)} = \xi^{(\ell)}$ and $a_\ell \leq 10 \log n$.
Again conditioned on event $\mathcal{E}$ (Part~3 of $\mathcal{E}$), $z^{(L)}$ obtained through Line~\ref{lin:nu_ell} to Line~\ref{lin:z_assign} in Algorithm~\ref{alg:fourier_sparse_recovery} satisfies 
\begin{align*}
\|\wh{x} - y^{(L-1)} - z^{(L)}\|_\infty \leq 2^{1-H} \nu_L = 2^{1-H} (2^{-(\log R^* - H +1)}\mu R^*) = \mu.
\end{align*}
Therefore, $y^{(L-1)} + z^{(L)}$ on Line~\ref{lin:return_recovered_signal} gives a desirable output.

\end{proof}

Now we present our main theorem, which proves the correctness of Algorithm~\ref{alg:fourier_sparse_recovery}, and shows its sample complexity and time complexity.

\begin{theorem}[main result, formal version]\label{thm:fourier_sparse_recovery_formal}
Let $n = p^d$ where both $p$ and $d$ are positive integers.
Let $x \in \C^{[p]^d}$.
Let $k \in \{1, \ldots, n\}$. 
Assume we know $\mu = \frac{1}{k} \|\wh{x}_{-k}\|_2$ and $R^* \geq \|\wh{x}\|_\infty / \mu$ where $\log R^* = O(\log n)$.
There is an algorithm (Algorithm~\ref{alg:fourier_sparse_recovery}) that takes $O(k \log k \log n)$ samples from $x$, runs in $O(n \log^3 n \log k)$ time, and outputs a $O(k)$-sparse vector $y$ such that
\begin{align*}
\| \wh{x} - y \|_{\infty} \leq \frac{1}{\sqrt{k}}\min_{k-\sparse~x'} \| \wh{x} - x' \|_2
\end{align*}
holds with probability at least $1 - 1 / \poly(n)$.
\end{theorem}

\begin{proof}
The correctness of Algorithm~\ref{alg:fourier_sparse_recovery} follows directly from Lemma~\ref{lem:E_happen_whp} and Lemma~\ref{lem:correctness_of_fourier_sparse_recovery}.
The number of samples from $x$ is 
\begin{align*}
B \cdot R \cdot H = O(k \cdot \log n \cdot \log k) = O(k \log k \log n).
\end{align*}
Its running time is dominated by $L \cdot H = O(\log k \log n)$ invocations of \textsc{LinfinityReduce} (in Algorithm~\ref{alg:linfinity_reduce}).
By Lemma~\ref{lem:linfinity_reduce}, the running time of \textsc{LinfinityReduce} is $O(n \log^2 n)$.
Therefore, the running time of Algorithm~\ref{alg:fourier_sparse_recovery} is 
\begin{align*}
O(L \cdot H \cdot n \log^2 n) = O( \log k \cdot \log n \cdot n \log^2 n ) = O( n \log^3 n \log k).
\end{align*}
\end{proof}





\end{document}